\documentclass[journal]{IEEEtran}
\usepackage[T1]{fontenc} 
\usepackage[utf8]{inputenc}
\usepackage[draft]{hyperref} % optional
\usepackage{setspace}
\usepackage{amsmath}
\usepackage{amssymb}
\usepackage{amsthm}
\usepackage{amsfonts}
\usepackage{color}
\usepackage{graphicx}
\usepackage{subfigure}  % for making a row or column of subfigures
\usepackage{setspace}
\usepackage{amsmath}
\usepackage{amssymb}
\usepackage{stfloats}
\usepackage{amsthm}
\usepackage{amsfonts}
\usepackage{color}
\usepackage{graphicx}
\usepackage{subfigure} 

\usepackage{multirow}

\usepackage{cite}
\usepackage{mathtools}
\usepackage{stmaryrd}
\usepackage{array}
\usepackage{mathrsfs}

\usepackage{soul}
\usepackage{enumitem}

\newtheorem{thm}{{Theorem}}

\usepackage{algorithm}
\usepackage{algpseudocode}

\algnewcommand{\algorithmicgoto}{\textbf{go to}}%
\algnewcommand{\Goto}[1]{\algorithmicgoto~\ref{#1}}%

\begin{document}
\author{   Eyar Azar, Satish~Mulleti, {\it Member, IEEE}, Yonina C. Eldar, {\it Fellow, IEEE}% <-this % stops a space

	\thanks{\scriptsize E. Azar and Y. C. Eldar are with the Faculty of Math and Computer Science, Weizmann Institute of Science, Israel. S. Mulleti is with the Department of Electrical Engineering, Indian Institute of Technology Bombay, Mumbai, India. Emails: eyar.azar@weizmann.ac.il, mulleti.satish@gmail.com, yonina.eldar@weizmann.ac.il}
	\thanks{\scriptsize {
    This research was partially supported by the Israeli Council for Higher Education (CHE) via the Weizmann Data Science Research Center, by a research grant from the Estate of Tully and Michele Plesser, by the European Union’s Horizon 2020 research and innovation program under grant No. 646804-ERC-COG-BNYQ, by the Israel Science Foundation under grant no. 0100101, and by the QuantERA grant  C'MON-QSENS.}\vspace{-0.3cm}}
}
\title{Robust Unlimited Sampling Beyond Modulo}
\markboth{Submitted to the IEEE Transactions on Signal Processing}%
{Shell \MakeLowercase{\textit{et al.}}: Bare Demo of IEEEtran.cls for Journals}
\maketitle

\begin{abstract}
Analog to digital converters (ADCs) act as a bridge between the analog and digital domains. Two important attributes of any ADC are sampling rate and its dynamic range. For bandlimited signals, the sampling should be above the Nyquist rate. It is also desired that the signals' dynamic range should be within that of the ADC's; otherwise, the signal will be clipped. Nonlinear operators such as modulo or companding can be used prior to sampling to avoid clipping. To recover the true signal from the samples of the nonlinear operator, either high sampling rates are required or strict constraints on the nonlinear operations are imposed, both of which are not desirable in practice. In this paper, we propose a generalized flexible nonlinear operator which is sampling efficient. Moreover, by carefully choosing its parameters, clipping, modulo, and companding can be seen as special cases of it. We show that bandlimited signals are uniquely identified from the nonlinear samples of the proposed operator when sampled above the Nyquist rate. Furthermore, we propose a robust algorithm to recover the true signal from the nonlinear samples. We show that our algorithm has the lowest mean-squared error while recovering the signal for a given sampling rate, noise level, and dynamic range of the compared to existing algorithms. Our results lead to less constrained hardware design to address the dynamic range issues while operating at the lowest rate possible. 
\end{abstract}

\IEEEpeerreviewmaketitle

\section{Introduction}
Sampling plays a crucial role in representing analog signals digitally and processing them efficiently using digital signal processors. Among different sampling techniques, the Shannon-Nyquist sampling framework is widely used. In this framework, bandlimited signals are represented by their instantaneous samples with sampling rate  greater than or equal to the Nyquist rate, which is twice the maximum frequency component. The cost and power consumption of an analog-to-digital converter (ADC) increase with an increase in the sampling rate. Hence, it is desirable to sample closer to the Nyquist rate of the signals.

Apart from sampling rate, another key attribute of an ADC is its dynamic range. Ideally, the dynamic range of an ADC should be larger than that of the input analog signal; otherwise, the signal gets clipped. Clipping is a nonlinear process that results in loss of information. Several approaches have been proposed to address clipping or the dynamic range issue. These approaches can be broadly divided into two categories based on whether preprocessing is applied before sampling. One of the techniques that does not involve a preprocessing step uses the fact that samples of bandlimited signals are correlated when measured above the Nyquist rate. Such correlation among samples are used to retrieve any missing samples \cite{marks_clipping2, marks_clipping1}. To address clipping, the signal is oversampled  and the clipped samples are considered as the missing ones and are then recovered from the remaining unclipped samples. However, theoretical guarantees are lacking for this approach.

An alternative is to use spectral holes in the analog signal. The problem of clipping is prevalent in multiband communication systems where several signals are simultaneously transmitted. This results in high-dynamic-range signals at the receiver and may result in clipping. 
In general, the received signal supposed to have spectral holes, due to to multiband nature. However, due to clipping, the received signal has wider bandwidth and does not have vacant bands.
% Due to the multiband nature, the received signal has spectral holes, whereas, the clipped signal has wider bandwidth compared to the received signal and does not have vacant bands. 
In \cite{abel_clipping, rietman_clipping}, information about the vacant bands is used to differentiate between the original signal and the clipped ones. In the aforementioned techniques, either large oversampling is required \cite{marks_clipping2, marks_clipping1} or prior knowledge of the vacant bands is needed \cite{abel_clipping, rietman_clipping}. In addition, there are no theoretical guarantees derived for these approaches.

Clipping can be avoided by using an attenuator. In this approach, oversampling is not required. However, natural signals typically consist of a few large-amplitude regions and several regions with low amplitudes. Attenuation may push the low amplitude signals below the noise floor. Attenuators with variable gains, such as automatic gain controls (AGCs) and companders, are used in communication applications to address the dynamic range issue without distorting the small amplitude regions. In AGC, a chain of amplifiers is used with a feedback loop such that a suitable output level is maintained at the output \cite{perez2011automatic,mercy1981review}. The AGC circuit uses a closed-loop feedback mechanism and maintaining stability of the circuit for different signal levels may be difficult.

Companding is an alternative, popular approach with variable gain where smaller amplitudes have larger gain compared to the larger ones. Similar to clipping, companding is a nonlinear operation and increases the bandwidth of the signal. Beurling proved that knowledge of the companded signal over the input signal's bandwidth is sufficient to uniquely identify the signal provided that the compander is a monotone function and its output to finite energy input has finite energy \cite{landau_compander}. This result implies that a companded signal can be sampled at the Nyquist rate by applying an antialiasing filter before sampling. Landau et al. \cite{landau_distorted_bl} proposed an iterative algorithm to recover a bandlimited signal from its companded and lowpass version. The algorithm converges to the true signal provided that the response of the compander is differentiable over the dynamic range of the input signal.

While the aforementioned companding-based approaches operate at a minimal possible sampling rate, the requirement of monotonicity, differentiability, and finite energy output limits their hardware implementation \cite{landau_compander}, \cite{landau_distorted_bl}. Specifically, it is difficult to realize a monotone operator over the entire dynamic range of the signal. An additional approach to companding is to use a modulo operation before sampling to restrict the dynamic range. Specifically, the input signal is folded back when it crosses the dynamic range of the ADC. Hardware realization of such high-dynamic-range ADCs, also known as \emph{self-reset} ADCs, are discussed in the context of imaging \cite{sradc_park, sradc_sasagwa, sradc_yuan, krishna2019unlimited}. Along with samples of the modulo signal, these architectures store side information such as the amount of folding for each sample, or, the sign of the folding. Measuring the side information leads to complex circuitry at the sampler but enables computationally simple recovery. 

Bhandari et al. considered \emph{unlimited sampling}, where the side information is not measured and only folded or modulo samples are used for recovery \cite{unlimited_sampling17, uls_tsp}. The authors showed that for bandlimited signals sampling higher than the Nyquist rate is sufficient to uniquely identify the signal from its modulo samples. An algorithm to determine the true or unfolded samples from the modulo ones is suggested by applying an extension of the Itoh's unwrapping algorithm \cite{itoh}. Specifically, the authors showed that by oversampling the bandlimited signals there exists a positive integer $N$ such that the modulo of the $N$-th order differences of the modulo samples is equal to the $N$-th order differences of the true samples. Once the higher-order differences of the true samples are computed, the true samples are recovered by applying $N$-th order summation. The existence of such $N$ is guaranteed provided that the sampling rate is greater than or equal to $(2\pi e)$-times the Nyquist rate where $e$ is the Euler's constant. That is, an oversampling factor (OF) 17-times is required \cite{unlimited_sampling17, uls_tsp}. In the presence of bounded noise, a much higher OF compared to $(2\pi e)$ is needed \cite{uls_tsp}. In addition, the recovery algorithm is sensitive to noise due to higher-order difference operations.

Romanov and Ordentlich \cite{uls_romonov} improved on the previous results and proposed an algorithm that requires the sampling rate to be slightly above the Nyquist rate. The authors leverage the fact that there exists a time instant beyond which the signal lies within the dynamic range of the ADC. From these unfolded samples, the folded samples are predicted by using the correlation among the samples. However, simulation results of the algorithm are not presented, especially, in the presence of noise. Gan and Liu \cite{uls_multichannel} considered a multichannel extension of modulo sampling. In the absence of noise, the authors showed that two channels, each of them operating at Nyquist rate, are sufficient to undo the modulo operation provided that the dynamic ranges of the two ADCs are coprime. The reconstruction is based on the application of the Chinese remainder theorem. Although perfect reconstruction is achieved by sampling at twice the Nyquist rate, coprime requirements on the dynamic ranges of the ADCs limits its practical application. Modulo sampling is also extended to different problems and signal models such as periodic bandlimited signals \cite{bhandari2021unlimited}, wavelets \cite{uls_wavelet}, mixture of sinusoids \cite{uls_sinmix}, finite-rate-of-innovation signals \cite{uls_fri}, multi-dimensional signals \cite{uls_md}, sparse vector recovery \cite{uls_gamp, uls_sparsevec}, direction of arrival estimation problem \cite{uls_doa}, computed tomography \cite{uls_radon}, and graph signals \cite{uls_graph}. In addition to theory and algorithms, hardware prototypes high-dynamic of range ADCs by using modulo operators are presented in \cite{bhandari2021unlimited,mod_demo,mod_demo2}.

In summary, AGC, companding, and modulo are different ways of addressing the dynamic range issues, there are several drawbacks such as missing theoretical guarantees, stability, requirements of smooth and monotone operators, and algorithms operating at higher sampling rates than the Nyquist rate. In addition, the solutions are developed independently and lack common recovery methods. A single reconstruction algorithm that can recover bandlimited signals from clipped, companded, or modulo samples robustly in the presence of noise and from the minimal sampling rate is lacking.

In this paper, we present a general framework to address dynamic range of the ADCs where all existing approaches can be treated as special cases and provide scope to design new amplitude limiters. Specifically, we consider a non-linear transformation function before sampling with the following desired response to a bandlimited input signal: (a) if the input signal is within the dynamic range of the ADC, then the response remains within the dynamic range and should be invertible; (b) for the part of the input signal beyond the dynamic range of the ADC, the output can take any arbitrary values. Invariability within the dynamic range of the ADC aids in achieving companding with any desired response and uniqueness when recovering the true samples from the non-linear samples. We derive theoretical guarantees and show that sampling above Nyquist rate is necessary and sufficient to recover the samples. Due to the generality of the operator, these guarantees apply to clipping (which were missing in previous works) and companding as well. 
% This also leads to theoretical guarantees for clipping 
% (which were missing in previous works). 
For companding, we do not require any smoothness constraints as in the previous results and hence a wider class of companders can be used.  

We then propose a sampling efficient and robust algorithm to recover the signal. Our algorithm uses the fact that the residual signal, the difference between the true signal and the output of the nonlinear operator, is time-limited for bandlimited signals. Hence, beyond the bandwidth of the signal, one can differentiate between the input signal and the residue. By oversampling the output of the nonlinear operator, we present an approach to recover the residual signal from the nonlinear samples. We show that the proposed algorithm can reconstruct signals from non-linearities such as clipping, modulo operation, and companding. For modulo operation, we compare our algorithm with those in \cite{uls_tsp} and \cite{uls_romonov}. We show that for a given noise level and dynamic range of the ADC, our method can reconstructs the signal for a lower sampling rate in comparison with the existing approaches.

The paper is organized as follows. In Section~\ref{sec:problem_formulation}, we define the class of generalized, non-linear operators considered in this paper and present the problem formulation. Identifiability results are derived in Section~\ref{sec:theory}.
In Section ~\ref{sec:Algo}, we present the proposed algorithm. Simulation results are provided in Section~\ref{sec:simulation} followed by conclusions. 
\begin{figure*}
    \centering
    {\includegraphics[width=6.5 in]{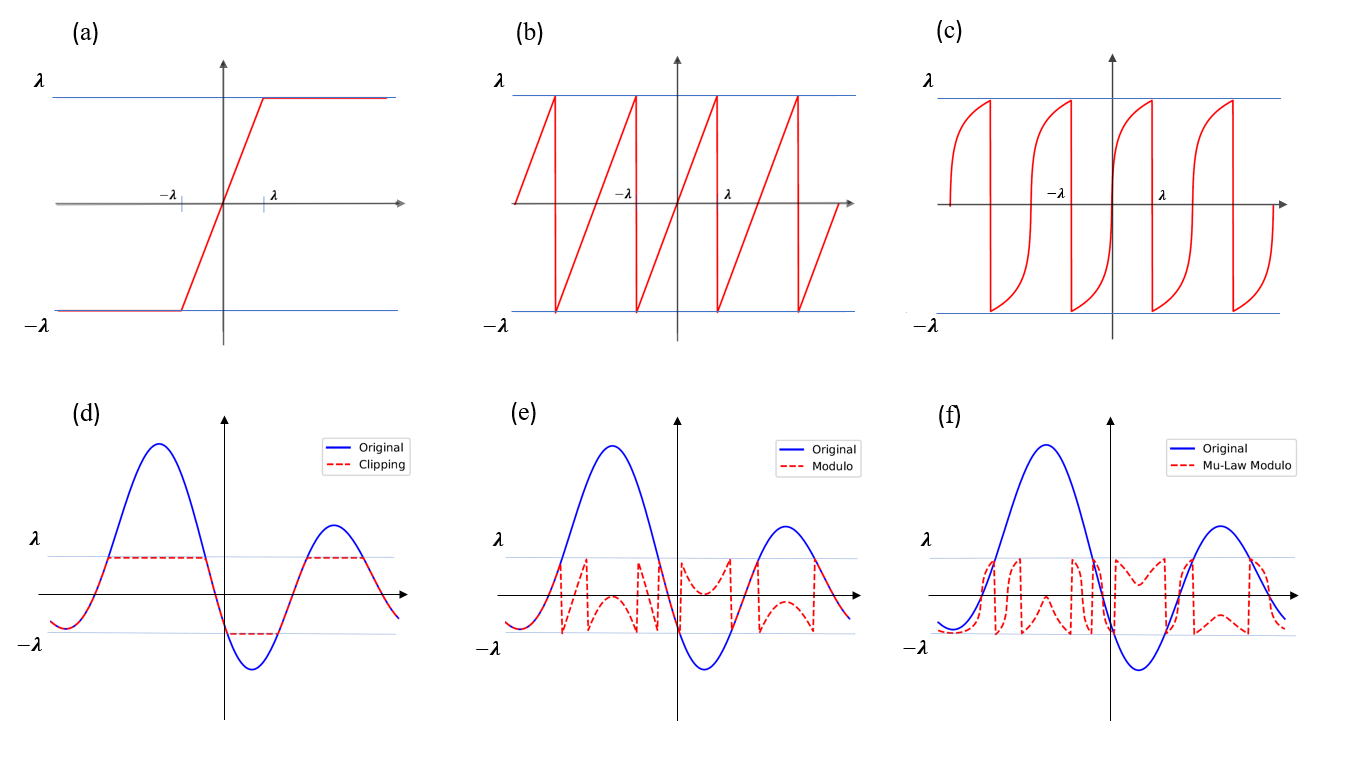}} 
    \caption{Examples of three non-linear functions:(a) Clipping as in \eqref{eq:clipping} (b) $\lambda$-modulo (c) $\mu$-law modulo (d) A bandlimited signal and its clipped version; (e) A bandlimited signal and output of modulo nonlinear operator; and (f) A bandlimited signal and output of a $\mu$-law modulo operator. }
    \label{fig:Non linear functions}
\end{figure*}
\begin{figure}
    \centering
    \includegraphics[width = 2.1in]{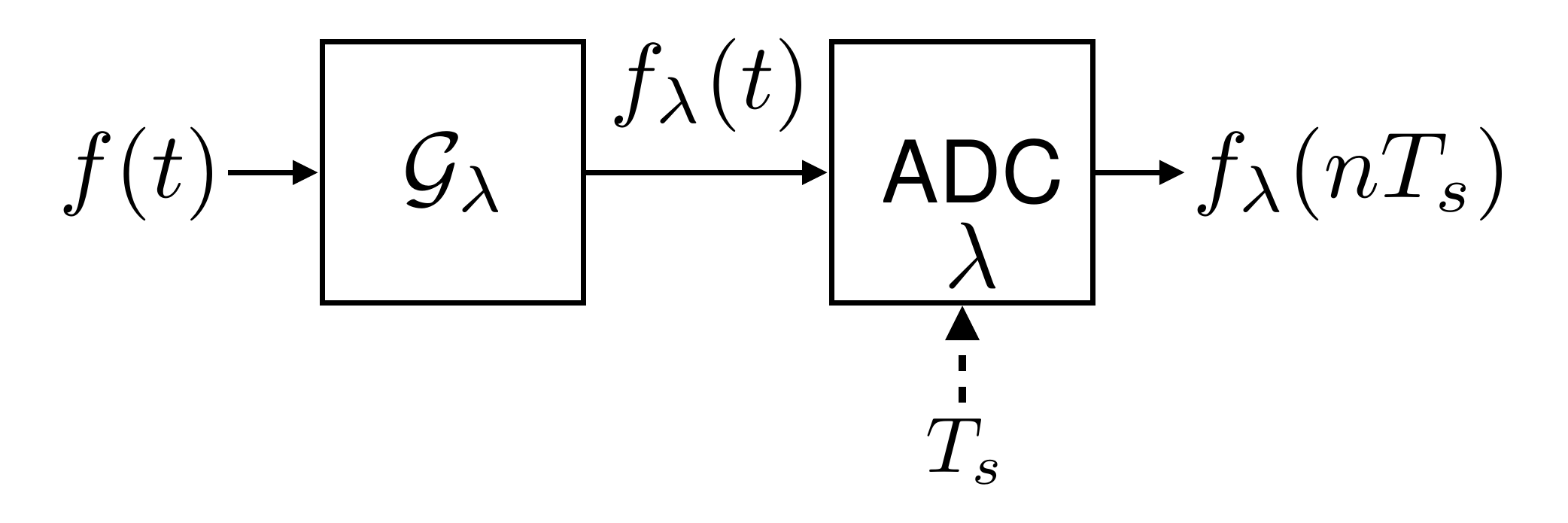}
    \caption{A schematic of generalized sampling: The bandlimited signal $f(t)$ is processed through a non-linear operator $\mathcal{G}_{\lambda}$ and then sampled by the ADC with a sampling interval $T_s$. The dynamic range of the ADC is $[-\lambda, \lambda]$.}
    \label{fig:sampling_block}
\end{figure}

We use the following notations and definitions in the paper. For a continuous-time analog signal $f(t)$ its Fourier transform is denoted as $F(\omega)$. Uniform samples of $f(t)$ are denoted by $f(nT_s)$ where $T_s >0$ is the sampling interval and $n \in \mathbb{Z}$. The corresponding sampling rate is $\omega_s = \frac{2\pi}{T_s}$ rads/sec. 
For a sequence $f(nT_s)$ its corresponding boldfaced symbol $\mathbf{f}$ denotes its vector form with $n$-th entry $\mathbf{f}[n] = f(nT_s)$.
The discrete-time Fourier transform (DTFT) is defined as
\begin{align}
    \label{eq:DTFT}
    \mathcal{F}\mathbf{f}= F(e^{\mathrm{j}\omega T_s}) = \sum_{n \in \mathbb{Z}}f(nT_s)\, e^{-\mathrm{j}\omega nT_s}.
\end{align}
For any interval $\rho \subset (-\omega_s/2,\,\omega_s/2)$, $\mathcal{F}_{\rho}\mathbf{f}$ denotes a partial DTFT $F(e^{\mathrm{j}\omega T_s})$ evaluated over $\omega \in \rho$, and  $\mathcal{F}^*_{\rho}$ denotes the adjoint operator of $\mathcal{F}_{\rho}$. Specifically, we have 
\begin{align}
    \label{eq: adjoint_partiel_DTFT}
    \mathcal{F}^*_{\rho}\mathcal{F}\mathbf{f}[n] = \frac{T_s}{2\pi}\int_\rho {F(e^{\mathrm{j}\omega T_s})\, e^{\mathrm{j}\omega n T_s} \mathrm{d}\omega}, \quad n\in \mathbb{Z}.
\end{align}
For any integer $N$, $\mathcal{S}_{N}$ denotes the space of sequences that have support over $\{-N, \cdots, N\}$, and $P_{\mathcal{S}_N}$ denotes the orthogonal projection onto the space ${\mathcal{S}_N}$ which sets all samples beyond $\{-N, \cdots, N\}$ to zero. 
The indicator function on domain $\mathcal{A}$ is denoted by $\mathbf{1}_{\mathcal{A}}(\cdot)$ .
The symbol $\mathrm{B}_{\omega_m}$ denotes the space of analog signals that are bandlimited to frequency interval $[-\omega_m, \omega_m]$. The sinc function is defined as $\text{sinc}(x) = \frac{\sin(\pi x)}{\pi x}$. For any two functions $g(t)$ and $f(t)$, a composite function is denoted as $g \circ f(t)$. For any $a \in \mathbb{R}$ and $\lambda \in \mathbb{R}^+$, the modulo operation $\mathcal{M}_{\lambda}(\cdot)$ is given as
	\begin{equation}
	\mathcal{M}_{\lambda}(a) = (a+\lambda)\,\, \text{mod}\,\, 2\lambda -\lambda.	
	\label{eq:mod_def}
	\end{equation}

% \begin{figure}
%     \centering
%  \includegraphics[width= 3.5 in]{figs/Original_mod_clip.pdf}
%     \caption{Caption}
%     \label{fig:my_label}
% \end{figure}

\section{Problem Statement}
\label{sec:problem_formulation}
\subsection{Preliminaries}

Consider a signal $f(t) \in \mathrm{B}_{\omega_m}$, an ADC with dynamic range $[-\lambda, \lambda]$, and sampling interval $T_s$. If the uniform samples $f(nT_s)$ are beyond the dynamic range of the ADC then they are clipped. Specifically, the output samples of the ADC $f_\lambda(nT_s)$ are given as
\begin{align}
    f_{\lambda}(nT_s) = \begin{cases}
    -\lambda, & f(nT_s) \leq -\lambda,\\
    f(nT_s), & |f(nT_s)| < \lambda,\\
    \lambda, & f(nT_s) \geq \lambda.
    \end{cases}
    \label{eq:clipping}
\end{align}

Clipping results in loss of information and generally requires high amount of oversampling to estimate $f(nT_s)$ from $f_\lambda(nT_s)$ for all $n \in \mathbb{Z}$ \cite{marks_clipping2, marks_clipping1}. To avoid clipping either instantaneous companding or modulo operations are used before sampling which limits the dynamic range of the signal. Instantaneous companding uses a nonlinear, monotone function $\mathcal{G}:\mathbb{R} \rightarrow \mathbb{R}$ such that $\mathcal{G}f(t) \in [-\lambda, \lambda]$ \cite{landau_compander}. One can recover $f(nT_s)$ from $\mathcal{G}f(nT_s)$ by sampling at the Nyquist rate. In addition, $\mathcal{G}$ boosts low amplitudes of the signal to improve the signal to noise ratio (SNR) which helps in accurate recovery. Existing companders are required to be monotone, differentiable, and $\mathcal{G}f(t) \in L^2(\mathbb{R})$ which limits their practical application \cite{landau_compander, landau_distorted_bl}. An alternative to avoid clipping is to perform a modulo operation prior to sampling, that is, sample $\mathcal{M}_{\lambda}f(t)$ instead of $f(t)$ \cite{uls_tsp}. As in companding we have that $\mathcal{M}_{\lambda}f(t) \in [-\lambda, \lambda]$. The existing algorithms to determine $f(nT_s)$ from $\mathcal{M}_{\lambda}f(nT_s)$, either operate at very high sampling rate \cite{uls_tsp} or are unstable in the presence of noise \cite{uls_romonov}. 
% Moreover, the modulo operation does not boost low frequency components as in the compander which improves the accuracy of signal reconstruction in the presence of noise. 

Our objective is to devise a non-linear operation, that has the advantages over the existing approaches such as companding and modulo, and existence of a robust practical recovery algorithm that operates at a rate closer to the Nyquist rate. 
To this end, we consider the following non-linear operator:
\begin{align}
\label{eq:non-linear-operator}
    \mathcal{G}_{\lambda}f(t) = \begin{cases}
    \text{arbitrary}, & |f(t)| > \lambda,\\
    g \circ f(t), & |f(t)| \leq \lambda,
    \end{cases}
\end{align}
where $g: [-\lambda, \lambda] \to [-\lambda, \lambda]$ is a known, memoryless, continuous, and invertible function. As we show later, our recovery does not depend on the response of the nonlinear operator for $|f(t)|>\lambda$ and hence we chose an arbitrary response. 

Both clipping and modulo operators are special cases of the operator $\mathcal{G}_{\lambda}$. To illustrate this,
three examples of $\mathcal{G}_{\lambda}$ are demonstrated in Fig.~\ref{fig:Non linear functions} together with their responses to a bandlimited signal. Fig.~\ref{fig:Non linear functions}(a) and Fig.~\ref{fig:Non linear functions}(b) illustrate the output of clipping (cf \eqref{eq:clipping}) and modulo $\mathcal{M}_{\lambda}$, respectively. In these two special cases, the function $g(t)$ is identity, that is, $g \circ f(t) = f(t)$, for $|f(t)|\leq \lambda$. For $|f(t)|>\lambda$, $\mathcal{G}_{\lambda}f(t) = \text{sgn}(f(t))\, \lambda$ for clipping and $\mathcal{G}_{\lambda}f(t) = \mathcal{M}_{\lambda}f(t)$ for modulo. Their outputs to a bandlimited signal are displayed in Fig.~\ref{fig:Non linear functions}(d) and Fig.~\ref{fig:Non linear functions}(e). Fig.~\ref{fig:Non linear functions}(c) shows the output of a operator consists of a $\mu$-law operator\footnote{A $\mu$-law operator is used for companding. Its response to a function $f(t)$ is given as $\displaystyle \text{sgn}(f(t)) \frac{\ln\left( 1+\mu |f(t)|/\|f(t)\|_{\infty}\right) }{\ln (1+\mu)}$ where $\mu >0$.} followed by modulo operations. Its output to a bandlimited signal is exhibited in Fig.~\ref{fig:Non linear functions}(f) where amplitudes closer to zero are amplified. Both clipping and modulo operation are well known in the literature. However, the $\mu$-law modulo operator is a novel one which is a combination of compander and modulo operators. These examples demonstrate that by careful selection of $g(t)$ and the response of the operator for $|f(t)|>\lambda$, different non-linear functions could be realized. 
In addition, it can be shown that companders, such as $\mu$-law and $A$-law, are special cases of the generalized operator.

\subsection{Problem Formulation}
Consider a bandlimited signal $f(t)$, non-linear operator, $\mathcal{G}_{\lambda}$, which operates on $f(t)$, and then sampled using an ADC with dynamic range $[-\lambda, \lambda]$ and sampling rate $\frac{1}{T_s}$.
% Consider a bandlimited signal $f(t)$ that is operated on by the nonlinear function $\mathcal{G}_{\lambda}$ and then sampled using an ADC with dynamic range $[-\lambda, \lambda]$ and sampling rate $\frac{1}{T_s}$.
The overall sampling scheme is shown in Fig. \ref{fig:sampling_block}.
Since the ADC clips signals beyond its dynamic range, we can assume that output of the operator is followed by a clipping operation prior to ADC. Hence, if we consider the response of the generalized operator together with the explicit clipping (due to ADC), we have that
\begin{align}
    |\mathcal{G}_{\lambda}f(t)| \leq \lambda.
    \label{eq:bounded_g}
\end{align}
In other words, the outcome of $G_{\lambda}f(t)$ is bounded to $[-\lambda, \lambda]$.

% We modify the response of the generalized operator in the presence of this explicit clipping as follows. The response of the operator for $|f(t)|>\lambda$ is

% \begin{align}
%     \mathcal{G}_{\lambda}f(t) = \lambda\,\text{sgn}(\mathcal{G}_{\lambda}f(t))  \,\, \text{if} \,\, |f(t)|>\lambda\,\, \text{and} \,\,|\mathcal{G}_{\lambda}f(t)|>\lambda.
%     \label{eq:boundedG}
% \end{align}
% The modification implies that if the response of the non-linear operator is beyond the dynamic range of the ADC for $|f(t)|>\lambda$ then it is clipped, otherwise, it remains arbitrary as in \eqref{eq:non-linear-operator}. Overall, the outcome of $G_{\lambda}f(t)$ is bounded to $[-\lambda, \lambda]$, that is, 
% \begin{align}
%     |\mathcal{G}_{\lambda}f(t)| \leq \lambda.
%     \label{eq:bounded_g}
% \end{align}

Our goal is to derive conditions on the sampling rate such that the signal $f(t)$ is uniquely identified from the \emph{non-linear} or \emph{folded} samples $f_{\lambda}(nT_s)$. Note that if one recovers the \emph{unfolded} or \emph{true} samples $f(nT_s)$ from $f_{\lambda}(nT_s)$ then $f(t)$ can be uniquely reconstructed from $f(nT_s)$ provided the sampling rate is greater than or equal to the Nyquist rate. In the rest of the discussion, we assume that there exist one or more samples such that $|f(nT_s)|> \lambda$. The assumption ensures that there are folded samples on which unfolding methods can be applied.

% Since the formulation is for the signals that are beyond the dynamic range of the signal, it seems logical to consider the signals with condition $\|f(t)\|_{\infty}>\lambda$. However, it is possible to construct a bandlimited signal even when $\|f(t)\|_{\infty}>\lambda$ provided that all the samples measured at the Nyquist rate are within the dynamic range of the ADC. Specifically, $\|f(t)\|_{\infty}>\lambda$ and $|f(nT_s)|\leq \lambda$ where $f(t)\in \mathrm{B}_{\omega_m}$ and $T_s = \frac{\pi}{\omega_m}$. In this scenario, samples are not clipped and perfect reconstruction is possible without a non-linear operations. To avoid such a scenario, we consider signals such that $|f(nT_s)|> \lambda$ for some samples and for the sampling intervals under consideration.

Non-linear operators prior to sampling as shown in Fig.~\ref{fig:sampling_block} are not new in the sampling literature. For example, Zhu \cite{Zhu_nonlinear} considers a non-linear sampling framework where an operator is used to convert an arbitrary signal to a bandlimited one. The framework enables the extension of the Shannon-Nyquist sampling framework to non-bandlimited functions. In another line of work, Dvorkind et al. \cite{tsvika_nonlinear} considered a non-linear operator together with a non-ideal sampling setup. The framework reflects the practical scenario where the measurement devices have inherent nonlinearities and they may not be measuring exact instantaneous (or ideal) samples. The authors derive conditions on the non-linearity and input signal model for perfect recovery. In addition, practical algorithms were discussed for signal reconstruction. Unlike these earlier works \cite{Zhu_nonlinear, tsvika_nonlinear} the operators considered in this work are not necessarily invertible.

In the next section, we derive identifiability results (which are independent of any recovery algorithm) for recovering bandlimited signals from samples of the non-linear operator. In Section ~\ref{sec:Algo} we present the proposed algorithms to recover the signal from minimal samples.

\section{Theoretical guarantees}
\label{sec:theory}
In this section, we derive necessary and sufficient conditions to uniquely identify a bandlimited function from its samples measured via the non-linear operator $\mathcal{G}_\lambda$. 

Our main results are summarized in the following theorem.
\begin{thm}[Identifiability conditions]
\label{theorem:necessary_identifiability}
    Consider the sampling scheme shown in Fig.~\ref{fig:sampling_block} where the operator $\mathcal{G_{\lambda}}$ is defined as in \eqref{eq:non-linear-operator}. Then any signal $f(t) \in L^2(\mathbb{R})\cap B_{\omega_{m}}$ is uniquely identifiable from its non-linear samples $\{f_{\lambda}(nT_s)\}$ iff sampling is performed above the Nyquist rate, that is, $T_s < \frac{\pi}{\omega_m}$.
    
    % with an input $f(t) \in L^2(\mathbb{R})\cap B_{\omega_{m}}$ such that $||f(n T_s)||_{\infty} > \lambda >0$. The operator $\mathcal{G_{\lambda}}$ is defined as in \eqref{eq:non-linear-operator}. To uniquely identify $f(t)$ from the samples $\{f_\lambda(n T_s)\}_{n \in \mathbb{Z}}$, it is necessary to sample at a rate greater than the Nyquist rate, that is, $\omega_s = \frac{2\pi}{T_s} > \omega_{\text{Nyq}}=2\omega_m$. 
\end{thm}

\begin{proof}
    See Appendix.
\end{proof}

Theorem~\ref{theorem:necessary_identifiability} implies that it is necessary and sufficient to sample above the Nyquist rate to uniquely identify a bandlimited signal from the samples of the non-linear operator $\mathcal{G}_\lambda$. Since the modulo operator $\mathcal{M}_{\lambda}$ is a special case of $\mathcal{G}_\lambda$, the result holds true for the modulo operator. Particularly, in terms of sampling rate, our sufficiency results are similar to that in \cite{uls_tsp} and hence Theorem~\ref{theorem:necessary_identifiability} is consistent with existing results. 
Note that, to the best of our knowledge, the necessary condition of sampling above the Nyquist rate has been proved for the first time in this work.
% In addition, we also show that sampling above the Nyquist rate is necessary. 
Our identifiability result depends only on the samples of $g \circ f(t)$ and not on the measurements of $\mathcal{G}_\lambda f(t)$ for $|f(t)| >\lambda$.

% that are within the dynamic range of the ADC and do not consider the output of $\mathcal{G}_{\lambda}$ beyond the dynamic range.

Since both clipping and companding are particular instances of the proposed operator, the results also hold for them. Hence, theoretically, it is possible to recover the clipped samples if they are measured over the Nyquist rate. In the case of companding, unlike Beurling's results \cite{landau_compander}, our guarantees do not require the operator to be smooth and hence extend the results to a broader class of companders.   
%   \begin{figure*}[!h]
%         \hrule
%         \begin{align}
%         \label{eq: Vandermonde}
%             \underbrace{\begin{pmatrix}
%             	e^{-\mathrm{j} \theta_0 (-N_{\lambda})} & e^{-\mathrm{j}\theta_0 (-N_{\lambda}+1)} & \hdots & e^{-\mathrm{j} \theta_0 (N_{\lambda})} \\
%                 e^{-\mathrm{j} \theta_1 (-N_{\lambda})} & e^{-\mathrm{j} \theta_1 (-N_{\lambda}+1)} & \hdots & e^{-\mathrm{j}\theta_1 (N_{\lambda})}\\
%                 \vdots & \vdots & \ddots & \vdots \\
%                 e^{-\mathrm{j} \theta_{2N_{\lambda}} (-N_{\lambda})} & e^{-\mathrm{j} \theta_{2N_{\lambda}} (-N_{\lambda}+1)} & \hdots & e^{-\mathrm{j}\theta_{2N_{\lambda}} (N_{\lambda})}
%             \end{pmatrix}}_{\mathbf{V}}
%             \underbrace{\begin{pmatrix}
%                 z(-N_{\lambda} T_s)\\
%                 z((-N_{\lambda}+1)T_s)\\
%                 \vdots \\
%                 z(N_{\lambda} T_s)
%             \end{pmatrix}}_{\mathbf{z}} = 
%             \underbrace{\begin{pmatrix} 
%                 -Y(e^{\mathrm{j}\theta_0})\\
%                 -Y(e^{\mathrm{j}\theta_1})\\
%                 \vdots \\
%                 -Y(e^{\mathrm{j}\theta_{2N_{\lambda}}})
%             \end{pmatrix}}_\mathbf{y}
%         \end{align}
%     \end{figure*} 
  
\section{A Robust and Lowrate Recovery Algorithm}
\label{sec:Algo}
    We next present an iterative algorithm for recovery of the samples $f(n T_s)$ from the non-linear samples $f_{\lambda}(n T_s) = \mathcal{G}_{\lambda} f(n T_s)$. The algorithm assumes that the sampling rate is greater than the Nyquist rate, that is, $\omega_s > 2 \omega_m$. The underlying principle for the algorithm is to use the out-of-band energy of the non-linear samples to reconstruct the residual signal. The residual signal is the difference between the true signal and non-linear ones. For this reason, we refer to the proposed algorithm as \emph{beyond bandwidth residual reconstruction ($B^2R^2$)}. For ease of discussion, we first present the algorithm for the modulo operator and then extend it to the general non-linear operator. The modulo setting is also considered in \cite{moduloicassp}.

    % The only condition that we require is that the sampling rate will be higher than the Nyquist rate, in other words $\omega_s > 2 \omega_m$. We show our method in the private case where first we consider a modulo non linear operator then we further expand it to the more general case where. 
    % In addition, in the next section, using simulations we show that our algorithm is more robust to noise than the other existing algorithms. 
    \subsection{$B^2R^2$ Algorithm for Modulo Operator}
        For $\mathcal{G}_{\lambda} = \mathcal{M}_{\lambda}$, the modulo samples are expressed as a linear combination of the true samples and a residual signal:
        \begin{align}
            f_\lambda(nT_s) = f(nT_s) + z(nT_s),
            \label{eq:resiudual}
        \end{align}
        where values of the residual sequence $z(nT_s)$ are integer multiples of $2\lambda$. Our approach is to first
        compute $z(nT_s)$ from the modulo samples $f_\lambda(nT_s)$ and then use \eqref{eq:resiudual} to determine $f(nT_s)$. To derive $z(nT_s)$ from $f_\lambda(nT_s)$, we use the following two properties of the finite energy bandlimited signals to separate $f(nT_s)$ from $f_\lambda(nT_s)$.
        \begin{itemize}
            \item Time-domain separation \cite{uls_romonov}: From the Riemann-Lebesgue Lemma it can be shown that $\lim_{|t| \to \infty}f(t) = 0$. This implies that for any $\lambda >0$ there exists an integer $N_{\lambda}$ such that $|f(nT_s)| < \lambda,$ for all $|n|> N_{\lambda}$. Hence, for $|n|> N_{\lambda}$, we have $f_\lambda(nT_s) = f(nT_s)$ and $z(nT_s) = 0$. Thus the modulo samples are equal to the true samples over a set of indices and they are used to distinguish the residual from the modulo samples in time. 
            \item Fourier-domain separation: Since the signal is sampled above the Nyquist rate,
            \begin{equation}
                F(e^{\mathrm{j}\omega T_s}) = 0, \quad \text{for} \quad  \omega_m < |\omega| < \omega_s / 2.
                \label{eq: zero_dtft0}
            \end{equation}
            By using the linearity of DTFT, from \eqref{eq:resiudual} we have that
            \begin{equation}
                F_\lambda(e^{\mathrm{j}\omega T_s}) = Z(e^{\mathrm{j}\omega T_s}), \quad \text{for} \quad  \omega_m < |\omega| < \omega_s / 2.
            \label{eq: zero_dtft1}
            \end{equation}
            This implies that one can differentiate the DTFT of the true samples and that of the residual by sampling above the Nyquist rate and looking beyond the bandwidth. 
        \end{itemize}
        
        In the rest of the discussion, we assume that $N_{\lambda}$ is known. From the time-domain separation, we infer that the residual signal has finite support on the integer set $\mathcal{N}_{\lambda} = \{-N_\lambda, \cdots, N_{\lambda}\}$. Combining the time-domain and the frequency-domain separations we arrive at the following relation:
        \begin{equation}
                F_{\lambda}(e^{\mathrm{j}\omega T_s}) = \sum_{n={-N_\lambda}}^{N_\lambda}z(n T_s)e^{-\mathrm{j} n T_s\omega},
                \label{eq:combine}
        \end{equation}
        for $\omega_m < |\omega| < \omega_s / 2$. Due to its finite support, the DTFT of $z(nT_s)$ is a trigonometric polynomial and it is given over an interval. 
        
        \subsubsection{A simple matrix-inversion-based solution}
        From \eqref{eq:combine} one can determine $z(nT_s)$ by sampling $F_{\lambda}(e^{\mathrm{j}\omega T_s})$ at $2N_\lambda +1$ points over the interval $\rho = (-\omega_s/2, -\omega_m) \cup (\omega_m, \omega_s/2)$ and inverting the resulting set of linear equations. The matrix that relates $z(nT_s)$ and samples of $F_{\lambda}(e^{\mathrm{j}\omega T_s})$ will have a Vandermonde structure with size of $(2N_\lambda +1) \times (2N_\lambda +1)$. The Vandermonde matrix is invertible if the $2N_\lambda +1$ points over the interval $\rho = (-\omega_s/2, -\omega_m) \cup (\omega_m, \omega_s/2)$ are unique. From the recovered residual signal, the true samples $f(nT_s)$ are determined by using \eqref{eq:resiudual}. In principle, the approach is similar to the algorithm proposed in \cite{bhandari2021unlimited} for periodic bandlimited signals. 
        
        Although the proposed approach looks simple, the matrix inversion used for estimating $z(nT_s)$ from the samples of $F_\lambda(e^{\mathrm{j}\omega T_s})$ may be unstable for large values of $N_\lambda$. To illustrate this consider $f(t) = \text{sinc}(t)$ where $\omega_m = 2\pi$. We consider its samples measured at a rate of $12 \omega_m$, that is, with an oversampling factor of 6. We used a modulo operator to limit the dynamic range before sampling and consider reconstruction by Vandermonde matrix inversion for $\lambda = 0.25$ and $\lambda = 0.2$. The true signals and the reconstructed signals are shown in Fig.~\ref{fig:Vand}. For $\lambda = 0.25$, $N_\lambda = 4$ and we observe perfect reconstruction, whereas, for $\lambda = 0.2$, $N_\lambda$ is 9 and perfect recovery is not achieved as shown in Fig.~\ref{fig:Van_lambda020}. In the following, we discuss an iterative algorithm that does not require matrix inversion and can reconstruct signals for larger values of $N_\lambda$. 
        \begin{figure}[!t]
    	    \begin{center}
        		\begin{tabular}{cc}
        			\subfigure[$\lambda = 0.25, N_\lambda = 4 $]{\includegraphics[width=1.6in]{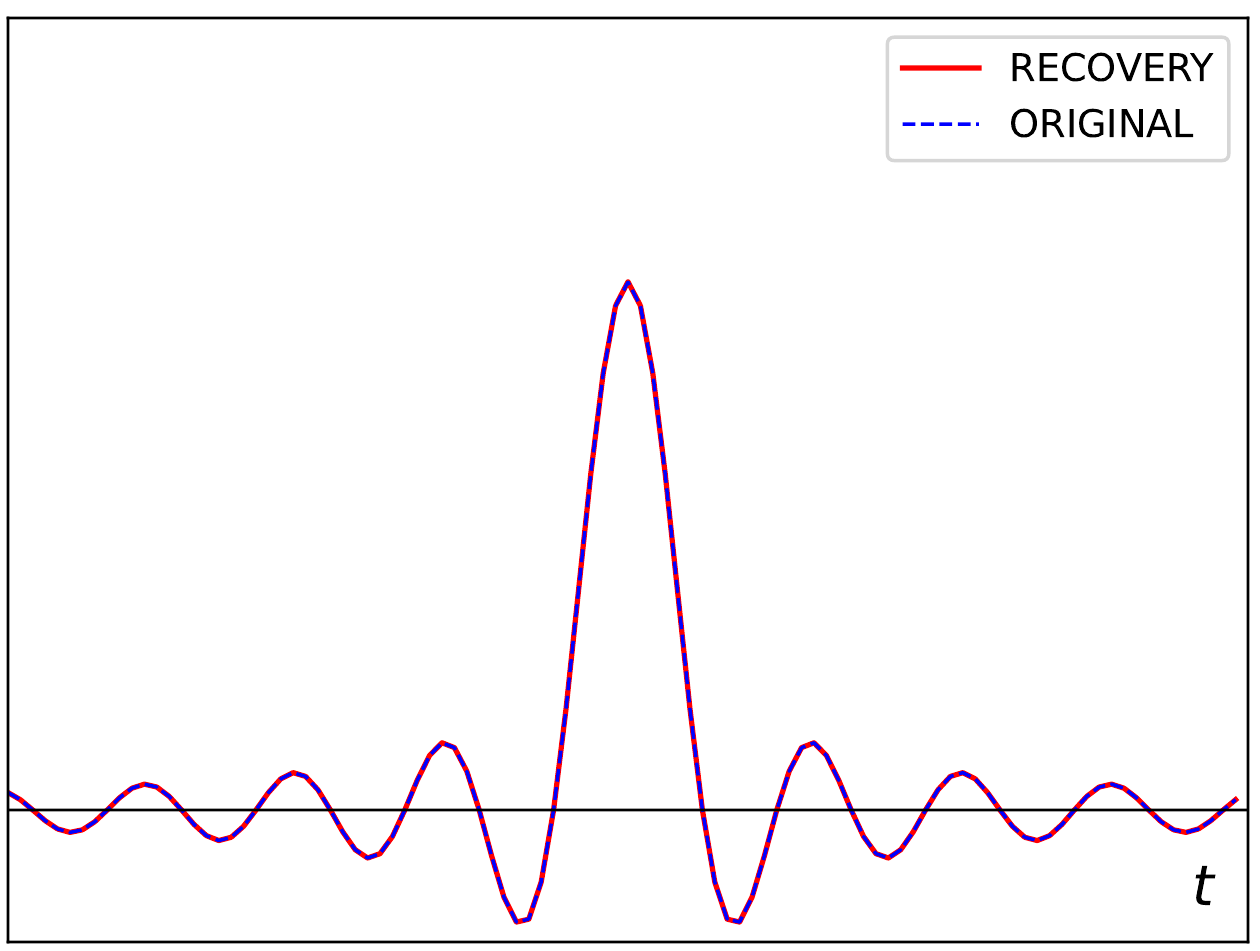}\label{fig:Van_lambda025}} 
        			\subfigure[$\lambda = 0.20, N_\lambda = 9 $]{\includegraphics[width=1.6in]{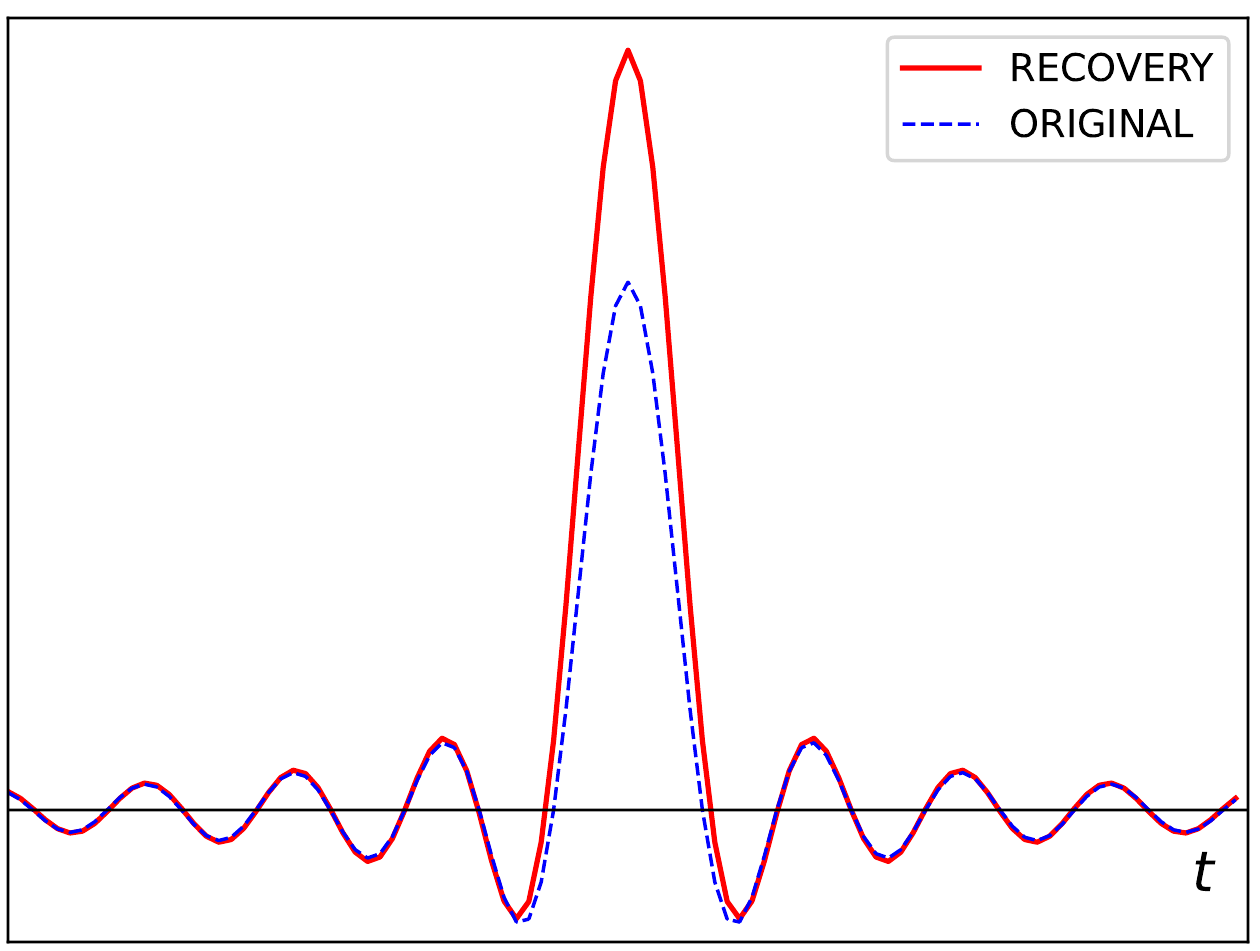}\label{fig:Van_lambda020}} 
        			\end{tabular}
        		%\end{array}
        		\caption{Reconstruction of bandlimited signals from its modulo samples by using Vandermonde inverse: (a) Perfect reconstruction and (b) Imperfect recovery. As $N_\lambda$ increases, the matrix inversion becomes unstable and perfect recovery is not achieved.}
        		\label{fig:Vand}
        	\end{center}
        \end{figure}

        % \begin{figure}[!h]
        %         \centering
        %         \begin{subfigure}[]
        %             \centering
        %             {\includegraphics[width=1.65 in]{figs/Journal/Examples/BBRR vs Vander/vander_0.25.pdf}}
        %         \end{subfigure}
        %         \begin{subfigure}[]
        %             \centering
        %             {\includegraphics[width=1.65 in]{figs/Journal/Examples/BBRR vs Vander/bbrr_0.25.pdf}}
        %         \end{subfigure}
        %         \vskip\baselineskip
        %         \begin{subfigure}[]
        %             \centering
        %             {\includegraphics[width=1.65 in]{figs/Journal/Examples/BBRR vs Vander/vander_0.2.pdf}}
        %         \end{subfigure}
        %         \begin{subfigure}[]
        %             \centering
        %             {\includegraphics[width=1.65 in]{figs/Journal/Examples/BBRR vs Vander/bbrr_0.2.pdf}}
        %         \end{subfigure}
        %         \caption{Comparison between invert a vandermonde matrix and $B^2R^2$ algorithm in recovering bandlimited signals from the modulo samples: Top row shows original signals and recovery for $\lambda = 0.25$ of (a) Vandermonde , (b) $B^2R^2$; Bottom row shows original signals and recovery for $\lambda  = 0.2$ of (c) Vandermonde , (d) $B^2R^2$. Perfect recovery is achieved for $\lambda = 0.25$ with both methods; for $\lambda = 0.2$ $B^2R^2$ achived perfect recovery while invert a Vandermonde matrix fails.}
        % \label{fig:vander_bbrr}
        % \end{figure}

        \subsubsection{$B^2R^2$: An iterative, optimization-based, computationally efficient solution}
        Here we propose an iterative algorithm that does not require any matrix inversion. The iterative algorithm is a solution to an optimization problem as discussed in the following. By using the operator and vector notations we rewrite \eqref{eq: zero_dtft1} as 
        \begin{align}
            \mathcal{F}_{\rho}\mathbf{f}_{\lambda} =  \mathcal{F}_{\rho}\mathbf{z},
        \label{eq:equal_vector}
        \end{align}
        where $\rho = (-\omega_s/2, -\omega_m) \cup (\omega_m, \omega_s/2)$. Since $z(nT_s)$ is time-limited to $\mathcal{N}_{\lambda}$, we have that 
        \begin{align}
            \mathbf{z} \in \mathcal{S}_{N_\lambda}.
       \label{eq:constarint}
        \end{align}
        Given the data-fitting term in \eqref{eq:equal_vector} and support constraint, recovery of $\mathbf{z}$ can be written as the following optimization problem:
        \begin{align}
            \underset{\mathbf{z}}{\min} \quad \mathrm{C}(\mathbf{z}) = \dfrac{1}{2}\|\mathcal{F}_{\rho}\mathbf{f}_{\lambda} - \mathcal{F}_{\rho}\mathbf{z}\|^2 \quad \text{s.t.} \quad \mathbf{z} \in \mathcal{S}_{N_\lambda}.
        \label{eq:opt2}
        \end{align}
    
        Problem \eqref{eq:opt2} can be solved using a projected gradient descent (PGD) method where at each iteration the solution iterates towards the negative gradient of the cost $\mathrm{C}(\mathbf{z})$ and is then projected onto the space $\mathcal{S}_{N_\lambda}$. In summary, starting from an initial point $\mathbf{z}^0 \in \mathcal{S}_{N_\lambda}$, the steps at the $k$-th iteration are
        \begin{equation}
            \begin{split}
                \mathbf{y}^{k} &= \mathbf{z}^{k-1} - \gamma_k \nabla \mathrm{C}(\mathbf{z}^{k-1}) \\
                \mathbf{z}^k &= P_{\mathcal{S}_{N_{\lambda}}}(\mathbf{y}^k),
            \end{split}
         \label{eq:pgd}
        \end{equation}
        where $\gamma_k >0$ is a suitable step-size, $\nabla \mathrm{C}(\mathbf{z}) = \mathcal{F}^*_{\rho}\mathcal{F}_{\rho}{(\mathbf{z} - \mathbf{f}_{\lambda})}$ is the gradient of $\mathrm{C(\mathbf{z})}$ and $P_{S_{N_\lambda}}(\cdot)$ is the orthogonal projection onto $\mathcal{S}_{N_{\lambda}}$. The operator $\mathcal{F}^*_{\rho}\mathcal{F}_{\rho}$ is a highpass operation. The sequence $\mathcal{F}^*_{\rho}\mathcal{F}_{\rho}(\mathbf{z} - \mathbf{f}_{\lambda})$ can be computed by filtering the sequence $\mathbf{z} - \mathbf{f}_{\lambda}$ with an ideal highpass filter with spectral support over $\rho$. Both the steps \eqref{eq:pgd} do not require any matrix inversion and hence instability and computational infeasibility for large $N_\lambda$ do not arise.

        In the case of modulo operation, the residual signal has an additional structure that every element of $\mathbf{z}$ is in $2\lambda \mathbb{Z}$. This constraint can be used after the support constraint in each step of the algorithm.
        %  \textbf{{\color{red} ADD PROOF} }
        
        We observed that, the rounding operation followed by PGD gives a good recovery of $\mathbf{z}$ from the modulo samples for small $N_{\lambda}$, whereas, for large $N_{\lambda}$ the estimation 
        tends to be more accurate at the edges of the support. 
        % Specifically, we observe that the first and last samples of $\mathbf{z}$, that is,  $\mathbf{z}[-N_\lambda]$ and $\mathbf{z}[N_\lambda ]$ are perfectly recovered for large $N_{\lambda}$. 
        Using this observation, we propose a sequential approach to improve the accuracy of estimation of the remaining samples.
        Starting from a given $N_\lambda$, let the PGD algorithm estimate of $\mathbf{z}$ be $\mathbf{\hat{z}}$. The estimate has support over $\mathcal{N}_\lambda$ and its values are integer multiples of $2\lambda$. In the absence of noise, 
        \begin{align}
            \mathbf{z}[n] = \mathbf{\hat{z}}[n], \quad n = \pm N_\lambda.
            \label{eq:edge_samples}
        \end{align}
        To estimate the remaining samples of $\mathbf{z}$ accurately, we define another sequence as 
        \begin{align}
          \label{eq:rec_signal}
          \mathbf{\hat{f}} = \mathbf{f_\lambda} - \mathbf{\hat{z}}.
        \end{align}
        Combining \eqref{eq:resiudual} and \eqref{eq:rec_signal},
        \begin{align}
        \label{eq:new_prob}
           \mathbf{\hat{f}}  = \mathbf{f} +\mathbf{z} - \mathbf{\hat{z}}.
        \end{align}
        From \eqref{eq:edge_samples} and \eqref{eq:new_prob} we have that $\mathbf{\hat{f}}[n] = \mathbf{f}[n],\, |n|>N_\lambda-1$. As a result, the new residual sequence $\mathbf{z} - \mathbf{\hat{z}}$ has support over $\{-(N_\lambda-1), \cdots, (N_\lambda-1)\}$, that is, $\mathbf{z} - \mathbf{\hat{z}} \in \mathcal{S}_{N_\lambda -1}$ and $\mathbf{z} - \mathbf{\hat{z}} \in 2\lambda \mathbb{Z}$. Hence, $\mathbf{\hat{f}}$ has a similar decomposition as in \eqref{eq:resiudual} except for the fact that its values need not be in the range $[-\lambda, \lambda]$. Despite that, we can redefine the optimization problem as in \eqref{eq:opt2} to estimate $\mathbf{z} - \mathbf{\hat{z}}$ from $\mathbf{\hat{f}}$ and use the PGD iterations as in \eqref{eq:pgd} to solve it. The residue $\mathbf{z} - \mathbf{\hat{z}}$ is correctly estimated for $n = \pm (N_\lambda-1)$, from which $\mathbf{f}$ can be determined at those locations. The process is repeated until all the samples are estimated. The algorithm, refereed as $B^2R^2$, is summarized in Algorithm~\ref{alg:algorithm_new}. For initialization, one can set $\mathbf{z}^0$ as $P_{\mathcal{S}_{N_{\lambda}}}\{\mathcal{F}^*_{\rho}\mathcal{F}_{\rho} \mathbf{f}_{\lambda}\}$. This is inverse-partial DTFT of $\mathcal{F}_\rho\mathbf{z}$ and we found that it serves as a good initial point. To illustrate this, we consider the same example as shown in Fig \ref{fig:Vand}. We observe that, unlike the matrix inversion method, the $B^2R^2$ algorithm achieves perfect reconstruction for $\lambda = 0.25, 0.20$ as shown in Fig \ref{fig:bbrr}.

        The proposed algorithm (Algorithm~\ref{alg:algorithm_new}) uses time-domain separation and frequency-domain separation properties to determine the residual signal. Whereas, the algorithm proposed in \cite{uls_romonov} uses these separation properties to directly predict the true samples from the folded ones. Specifically, the samples $f(nT_s), $ for all $ |n|\leq N_\lambda$ are predicted from $f(nT_s), $ for all $ |n| > N_\lambda$. Hence, both algorithms are entirely different although they use the same properties.
        \begin{figure}
        \begin{center}
		\begin{tabular}{cc}
			\subfigure[$\lambda = 0.25, N_\lambda = 4 $]{\includegraphics[width=1.6in]{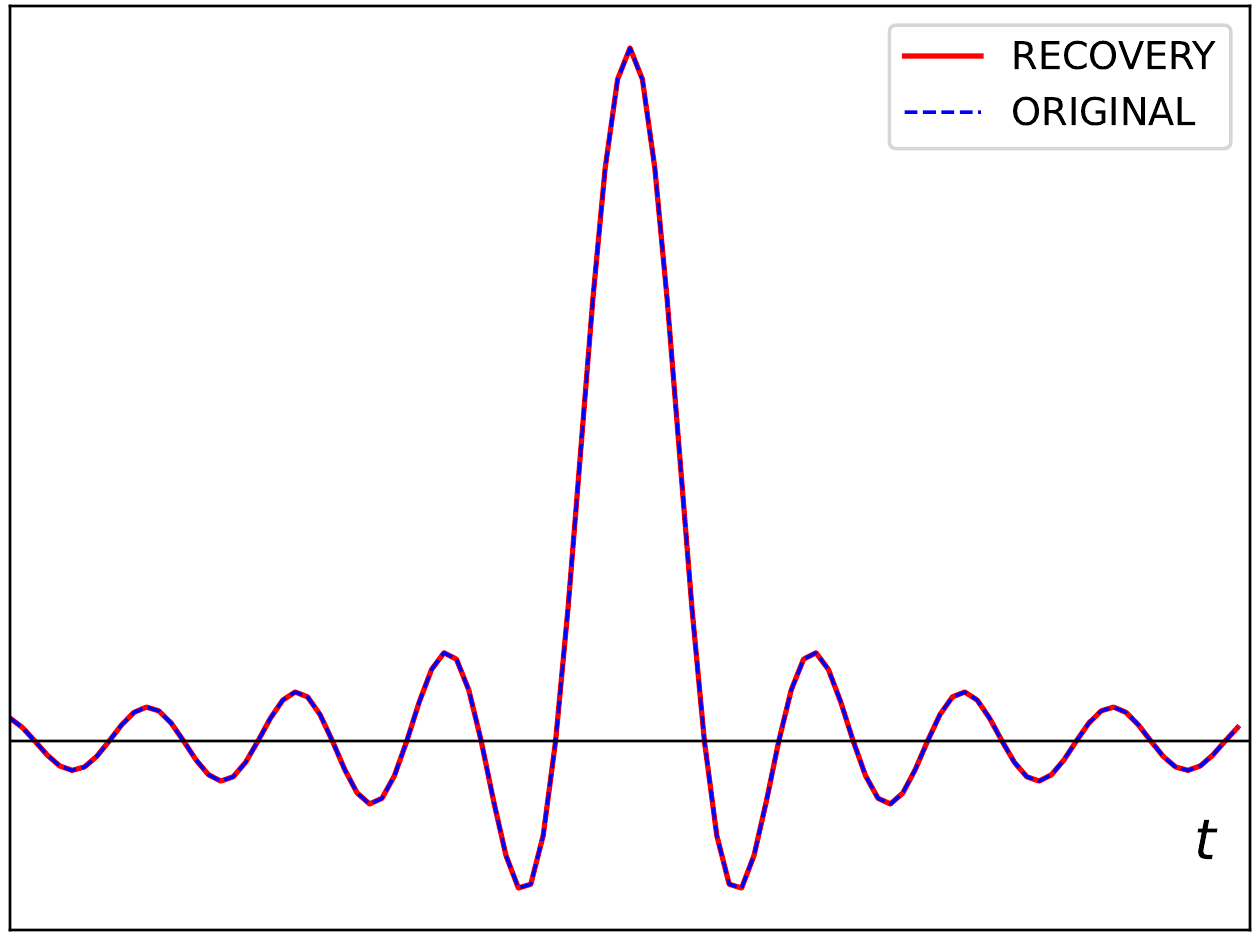}\label{fig:bbrr_lambda025}} 
			\subfigure[$\lambda = 0.20, N_\lambda = 9 $]{\includegraphics[width=1.6in]{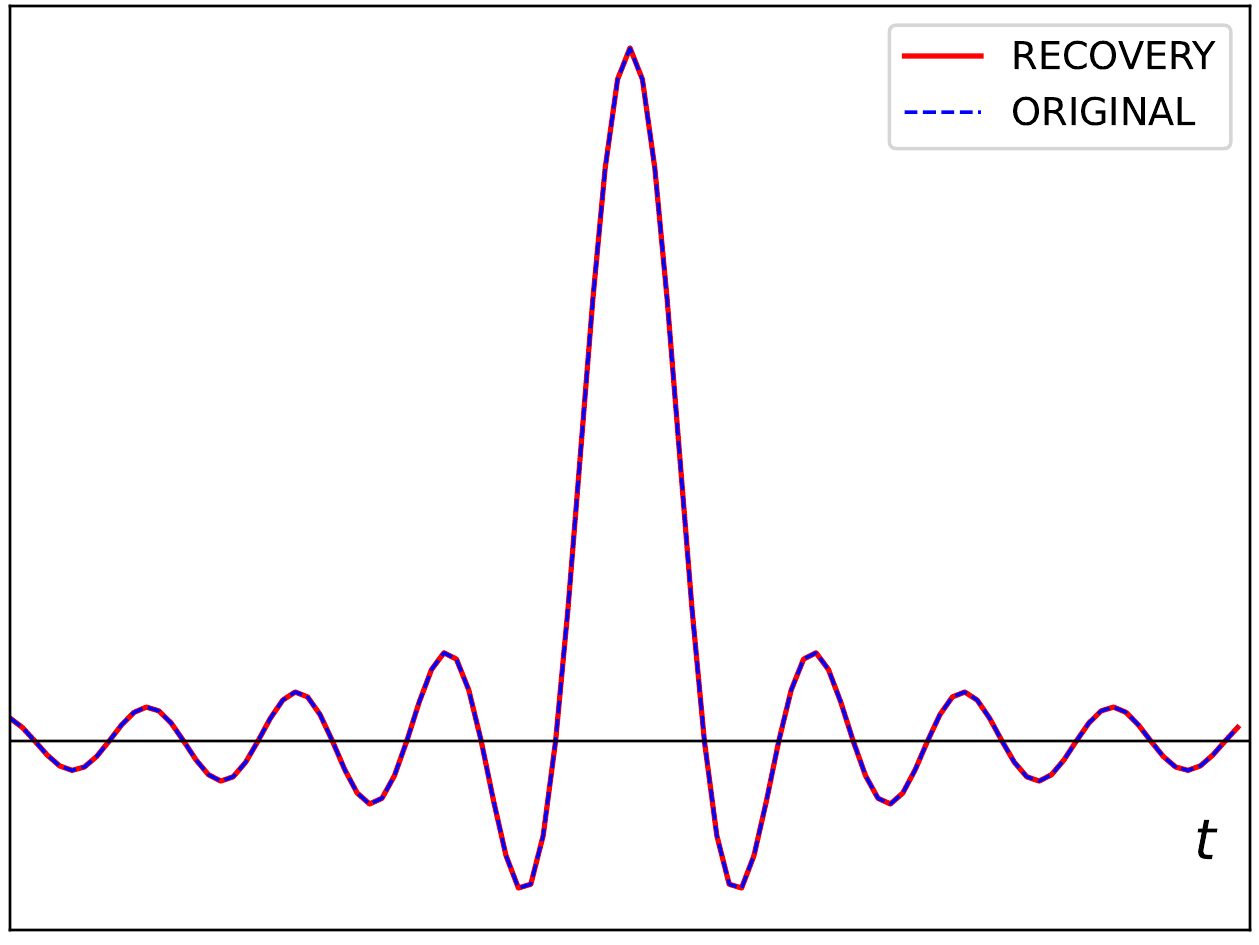}\label{fig:bbrr_lambda020}} 
			\end{tabular}
		%\end{array}
		\caption{Reconstruction of bandlimited signals from its modulo samples by using $B^2R^2$ algorithm: Perfect reconstruction is achieved for both $N_\lambda = 4$ and $N_\lambda = 9$.}
		\label{fig:bbrr}
	    \end{center}
	   \end{figure}

        \begin{algorithm}[!t]
            \caption{$B^2R^2$ for recovery of BL signals from modulo     samples}\label{alg:algorithm_new}
            \begin{algorithmic}[1]
                % \Procedure{Euclid}{$a,b$}\Comment{The g.c.d. of a and b}
                \State {\bf Input}$f_{\lambda}(n T_s)$ or $\mathbf{f}_{\lambda}$, $\lambda$, $\rho$ and $N_{\lambda}$
                \State \textbf{Intialize:} $\mathbf{\hat{f}} = \mathbf{f_\lambda}$, $\mathbf{z}^0 \in S_{N_\lambda}$ 
                \While{$N_\lambda >0$} 
                     \For{$k = 1$, $k{+}{+}$, Until stopping criteria}
                        \State Choose step size $\gamma_k$ by backtracking line search
                        \State $\mathbf{y}^k = \mathbf{z}^{k-1} - \gamma_k \mathcal{F^*}_{\rho}\mathcal{F}_{\rho}{(\mathbf{z}^{k-1} -\mathbf{\hat{f}})}$
                        \State $\mathbf{z}^k = P_{S_{N_\lambda}}(\mathbf{y^k})$
                    \EndFor
                    \State $\mathbf{\hat{z}} = \mathbf{z}^k$, \Comment{Estimation after applying PGD algorithm}
                    \State $\mathbf{\hat{z}} \leftarrow \left \lceil \frac{\lfloor \mathbf{\hat{z}} / \lambda \rfloor}{2} \right \rceil $  \Comment{rounding to $2\lambda \mathbb{Z}$},
                    \State $\mathbf{\hat{f}} \leftarrow \mathbf{\hat{f}} - \mathbf{\hat{z}}$
                    \State $N_{\lambda} \leftarrow N_{\lambda}-1$
                    \State $\mathbf{z}^0 =  P_{S_{N_\lambda}}(\mathbf{\hat{z}})$ 
              \EndWhile\label{euclidendwhile}
              \Statex \textbf{Output:}
                $\mathbf{f} = \mathbf{\hat{f}}$
          \end{algorithmic}
        \end{algorithm}
        
        % \begin{algorithm}[!t]
        %     \caption{$B^2R^2$ for recovery of BL signals from non-linear samples }\label{euclid_3}\label{alg:algorithm_general}
        %     \begin{algorithmic}[1]
        %         % \Procedure{Euclid}{$a,b$}\Comment{The g.c.d. of a and b}
        %         \State {\bf Input:} $f_{\lambda}(n T_s)$ or $\mathbf{f}_{\lambda}$, $\lambda$, $\rho$ and $N_{\lambda}$
        %         \State \textbf{Initialize:} $\mathbf{\hat{f}} = g^{-1} \circ \mathbf{f_\lambda}$, $\mathbf{z}^0 \in S_{N_\lambda}$ 
        %         \While{$N_\lambda >0$} 
        %              \For{$k = 1$, $k{+}{+}$, Until stopping criteria}
        %                 \State Choose step size $\gamma_k$ by backtracking line search
        %                 \State $\mathbf{y}^k = \mathbf{z}^{k-1} - \gamma_k \mathcal{F^*}_{\rho}\mathcal{F}_{\rho}{(\mathbf{z}^{k-1} -\mathbf{\hat{f}})}$
        %                 \State $\mathbf{z}^k = P_{S_{N_\lambda}}(\mathbf{y^k})$
        %             \EndFor
        %             \State $\mathbf{\hat{z}} = \mathbf{z}^k$, \Comment{Estimation after applying PGD algorithm}
        %             \State Project $\mathbf{\hat{z}}$ on to structure of $\mathbf{z}$.
        %             \State $\mathbf{\hat{f}} \leftarrow \mathbf{\hat{f}} - \mathbf{\hat{z}}$
        %             \State $N_{\lambda} \leftarrow N_{\lambda}-1$
        %             \State $\mathbf{z}^0 =  P_{S_{N_\lambda}}(\mathbf{\hat{z}})$ 
        %       \EndWhile\label{euclidendwhile_3}
        %       \Statex \textbf{Output:}
        %         $\mathbf{f} = \mathbf{\hat{f}}$
        %   \end{algorithmic}
        % \end{algorithm}
     
      \begin{figure*}[!t]
        \centering
        {\includegraphics[width=7 in]{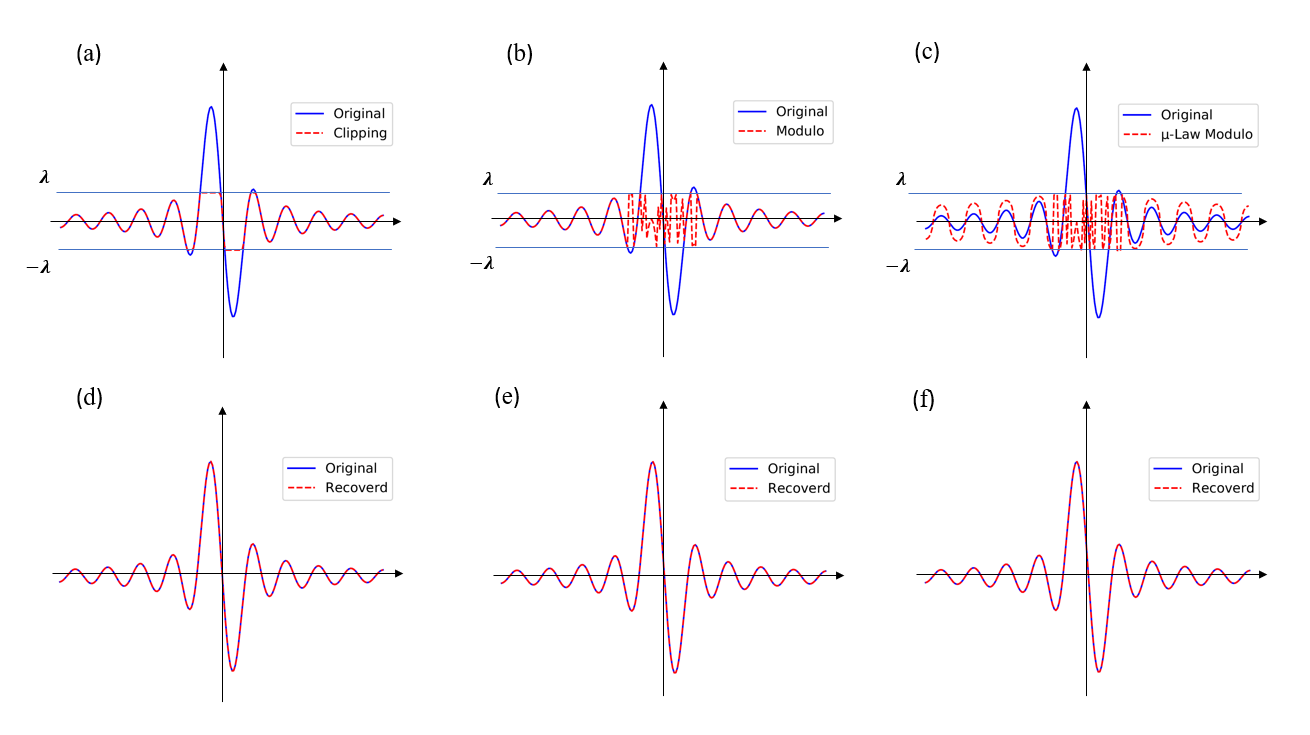}} 
        \caption{Reconstruction of bandlimited signals from non-linear samples by using $B^2R^2$ algorithm; Top row shows original signals and outputs of (a) clipping, (b) modulo operation, and (c) $\mu$-Law modulo; Bottom row shows recovery by using  $B^2R^2$ algorithm from samples of (d) clipping, (e) modulo operator, and (f) $\mu$-Law modulo recovery.}
    \label{fig:recovery_no_noise}
    \end{figure*}
    \subsection{$B^2R^2$ algorithm for general operator}    
         Here, we consider the general case when non-linear samples are given by the operator  $\mathcal{G}_{\lambda}$ as defined in \eqref{eq:non-linear-operator}. In this case, we define the residual signal as
         \begin{align}
         \label{eq:general residual}
            z(n T_s) = g^{-1} \circ f_{\lambda} (n T_s) - f(n T_s) = u(n T_s) - f(n T_s), 
         \end{align}
         where $u(n T_s) = g^{-1} \circ f_{\lambda} (n T_s)$.
         From the time and frequency separation we obtain that
         \begin{equation}
                U(e^{\mathrm{j}\omega T_s}) = \sum_{n={-N_\lambda}}^{N_\lambda}z(n T_s)e^{-\mathrm{j} n T_s\omega},
                \label{eq:combine_general}
        \end{equation}
        for $\omega_m < |\omega| < \omega_s / 2$.
        To estimate $z(nT_s)$ from $U(e^{\mathrm{j}\omega T_s})$ we consider an optimization framework as we did in the modulo case,
        \begin{align}
            \underset{\mathbf{z}}{\min} \quad \mathrm{C}(\mathbf{z}) = \dfrac{1}{2}\|\mathcal{F}_{\rho}\mathbf{u} - \mathcal{F}_{\rho}\mathbf{z}\|^2 \quad \text{s.t.} \quad \mathbf{z} \in \mathcal{S}_{N_\lambda}.
        \label{eq:opt2_general}
        \end{align}
    
        Problem \eqref{eq:opt2_general} can be solved using a PGD method as described in the modulo case. In summary, starting from an initial point $\mathbf{z}^0 \in \mathcal{S}_{N_\lambda}$, the steps at the $k$-th iteration are given in \eqref{eq:pgd}. 
        
        Although, most of the steps of the $B^2R^2$ algorithm for general operator remains same as in  Algorithm~\ref{alg:algorithm_new} but two steps make the difference. The first one is the initialization $\mathbf{\hat{f}}$. For general operator we initialize as
        \begin{align}
            \mathbf{\hat{f}} = g^{-1} \circ \mathbf{f_\lambda}.
            \label{eq:gen_init}
        \end{align}
        The second difference is in imposing structure of $\mathbf{z}$ to improve its accuracy. In Step-10 of Algorithm~\ref{alg:algorithm_new}, we use the fact that elements of $\mathbf{z}$ should be an integer multiple of $2\lambda$. Similarly, for different operators $\mathcal{G}_{\lambda}$ the residual signal $\mathbf{z}$ could have additional structure.

        For example, in clipping, if $f(nT_s)\geq \lambda$ then $f_\lambda(nT_s) = \lambda$. Hence, $z(nT_s) = f_\lambda(nT_s)- f(nT_s) \leq 0$. Similarly, $z(nT_s)\geq 0$ when $f_\lambda(nT_s) = -\lambda$. Hence sign of $z(nT_s)$ can be determined from the clipped samples. 
        % Next, we introduce a method to use this structure to improve the estimation of $z(nT_s)$. 
        To this end, by using the structure in the sign of the $z(nT_s)$ we can improve its estimation as
        % \begin{align}
            % \hat{z}(nT_s) &= \min(0, \hat{z}(nT_s)) \llparenthesis f_\lambda(nT_s) =\lambda \rrparenthesis \nonumber \\
            % &+ \max(0, \hat{z}(nT_s)) \llparenthesis f_\lambda(nT_s) = -\lambda \rrparenthesis \nonumber \\
            % &+ 0 \llparenthesis -\lambda <f_\lambda(nT_s) <\lambda \rrparenthesis.
            % \label{eq:clip_approx}
        % \end{align}
        
        \begin{align}
            \hat{z}(nT_s) &= \min(0, \hat{z}(nT_s)) \,  \mathbf{1}_{[f_\lambda(nT_s) =\lambda]}  \nonumber \\
            &+ \max(0, \hat{z}(nT_s)) \, \mathbf{1}_{[f_\lambda(nT_s) = -\lambda]}  \nonumber \\
            &+ 0 \, \mathbf{1}_{[ -\lambda <f_\lambda(nT_s) <\lambda]} .
            \label{eq:clip_approx}
        \end{align}
        For clipping, the three conditions $ f_\lambda(nT_s) =\lambda ,  f_\lambda(nT_s) = -\lambda $, and $ \lambda <f_\lambda(nT_s) <\lambda $ are mutually exclusive and one of them is always true for any folded sample. In \eqref{eq:clip_approx}, we set the values of $z(nT_s)$ to zero when the sign conditions are not satisfied. 
        Therefore, for clipping, we replace the operation in Step-10 of $B^2R^2$ algorithm with the approximation in  \eqref{eq:clip_approx}.
        % The approximation in  \eqref{eq:clip_approx} should be used instead of Step-10 of  Algorithm~\ref{alg:algorithm_new} for clipping.

        Hence in $B^2R^2$ algorithm for general operators we use the initialization in \eqref{eq:gen_init} together with suitable approximation in Step-10 should be used in Algorithm~\ref{alg:algorithm_new}. This generalized algorithm is designed for a general operator where clipping, companding, and modulo are special cases, it can recover bandlimited signals from samples of all these operators. 
        Importantly, the change in an ADC (together with the operator), does not require change in the algorithm.
        % Importantly, every time an ADC (together with the operator) changes, we need not change the algorithm.

\section{Simulations Results}
\label{sec:simulation}
    In this section, we present numerical results of different methods for recovering a bandlimited signal from the nonlinear samples. We first consider recovery in the absence of noise by using the proposed $B^2R^2$ algorithm. We then treat the noisy setting where we compare the proposed and existing approaches for reconstructing signals from modulo samples. We demonstrate the robustness to noise of the $B^2R^2$ algorithm for different parameters of $\lambda$ and the over-sampling factor.

    \subsection{Signal Reconstruction From Non-Linear Samples in the Absence of Noise }
    In this experiment, our goal is to demonstrate that the $B^2R^2$ algorithm perfectly reconstructs bandlimited signals from different nonlinear samples. Specifically, we consider the non-linearities discussed in Fig.~\ref{fig:Non linear functions}, namely, clipping (cf. \eqref{eq:clipping}), modulo operation as in \eqref{eq:mod_def}, and a $\mu$-law modulo operator. Let $\lambda  = 0.25$ for all these operators. Figs.~\ref{fig:recovery_no_noise}(a), (b), and (c) depict a bandlimited signal (in blue) and outputs of non-linear operators (in red). The true signals with corresponding recovered signals are shown in Figs.~\ref{fig:recovery_no_noise}(d), (e), and (f). For reconstruction from clipped samples, we used an $\text{OF} = 10$ and rest of the two operators $\text{OF} = 2$ was used. This shows that it is difficult to reconstruct from the clipped samples compared to modulo samples. Overall, we observe that the $B^2R^2$ algorithm recovers the original signal perfectly from the samples of different non-linear operators. 

    \begin{figure}[!h]
        \centering
        \begin{tabular}{c c}
          \subfigure[HOD ]{\includegraphics[width=1.6 in]{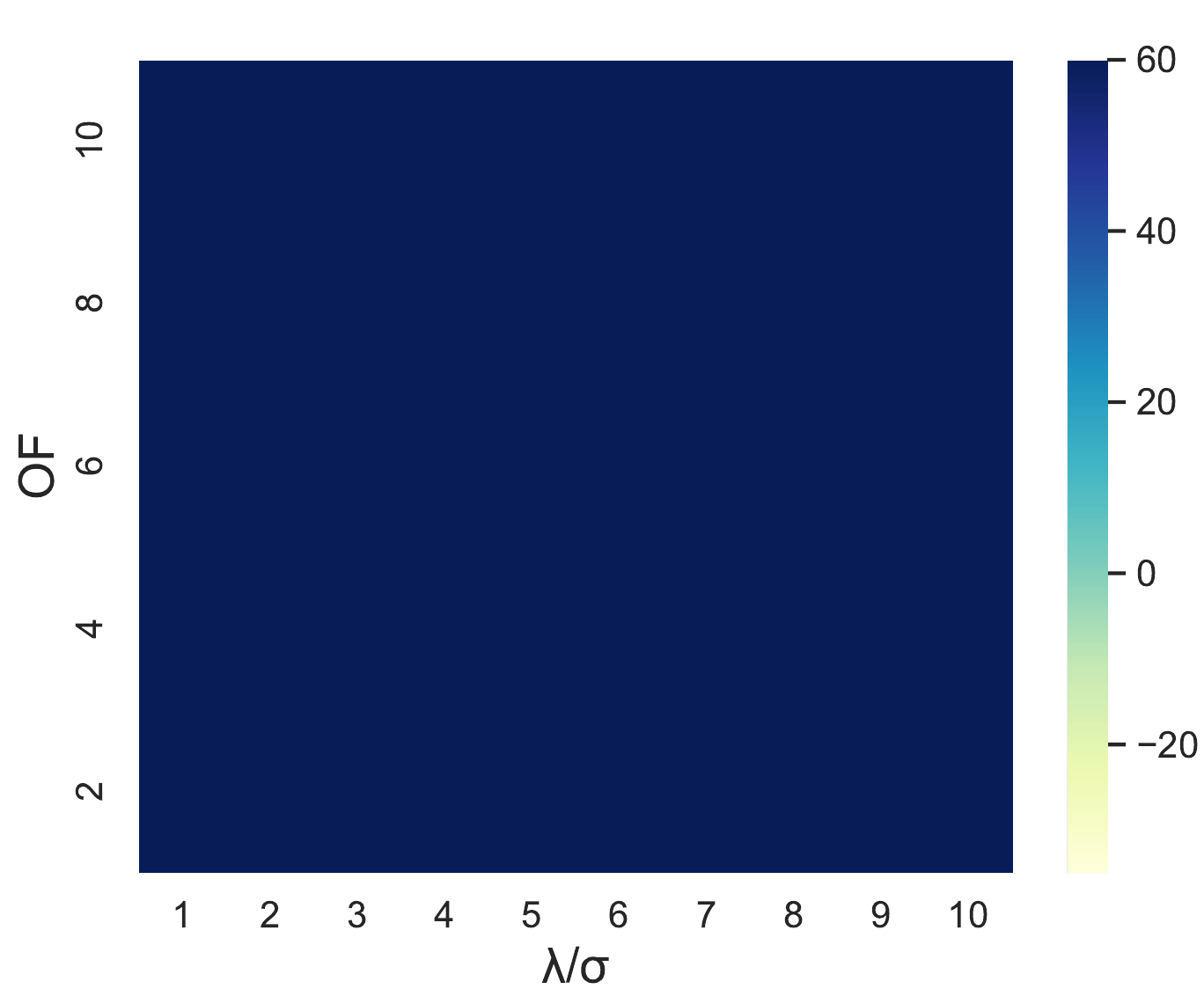}}   &  \subfigure[CPF]{\includegraphics[width=1.6 in]{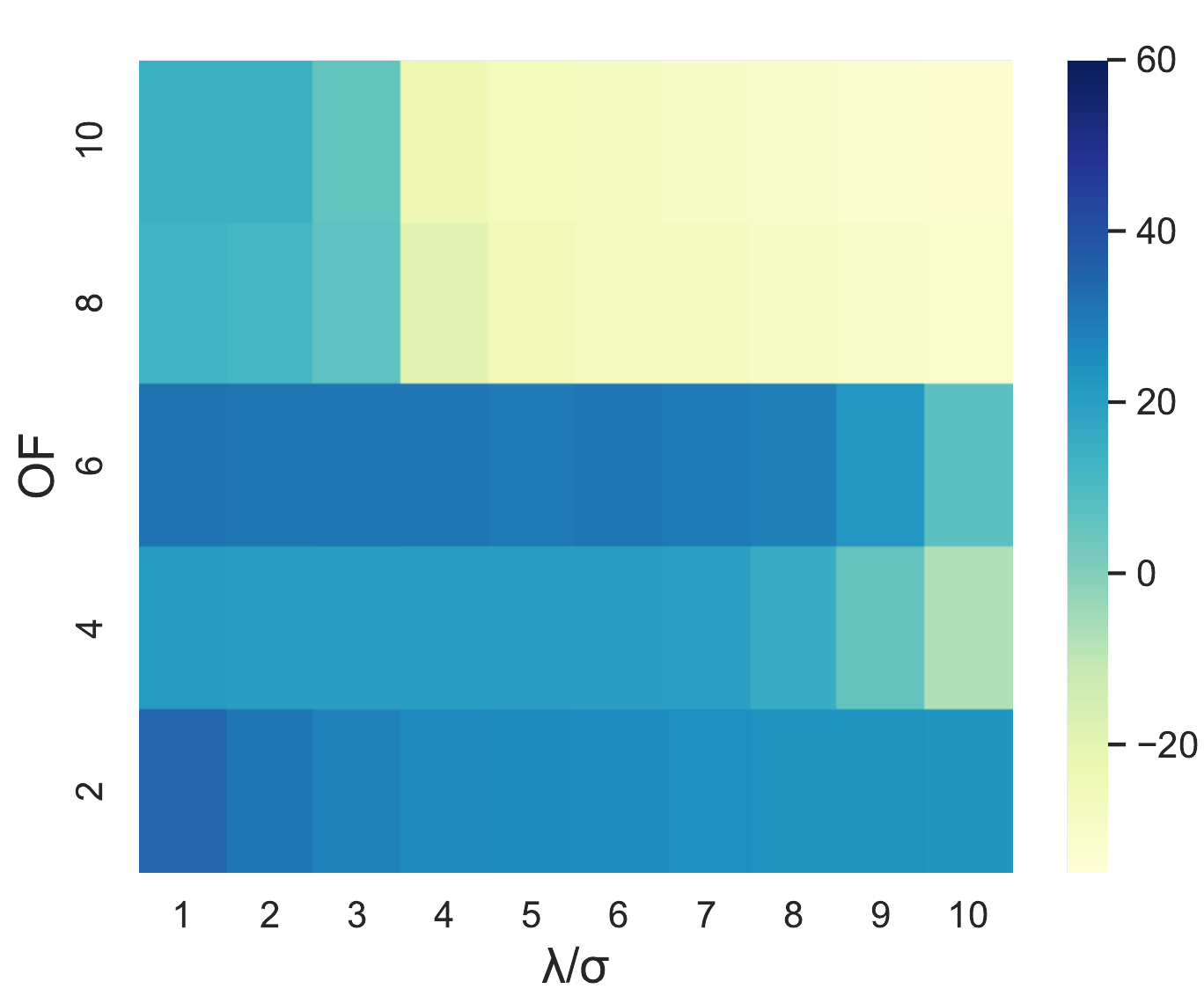}}
        \end{tabular}
        \subfigure[$B^2R^2$]{\includegraphics[width=1.6 in]{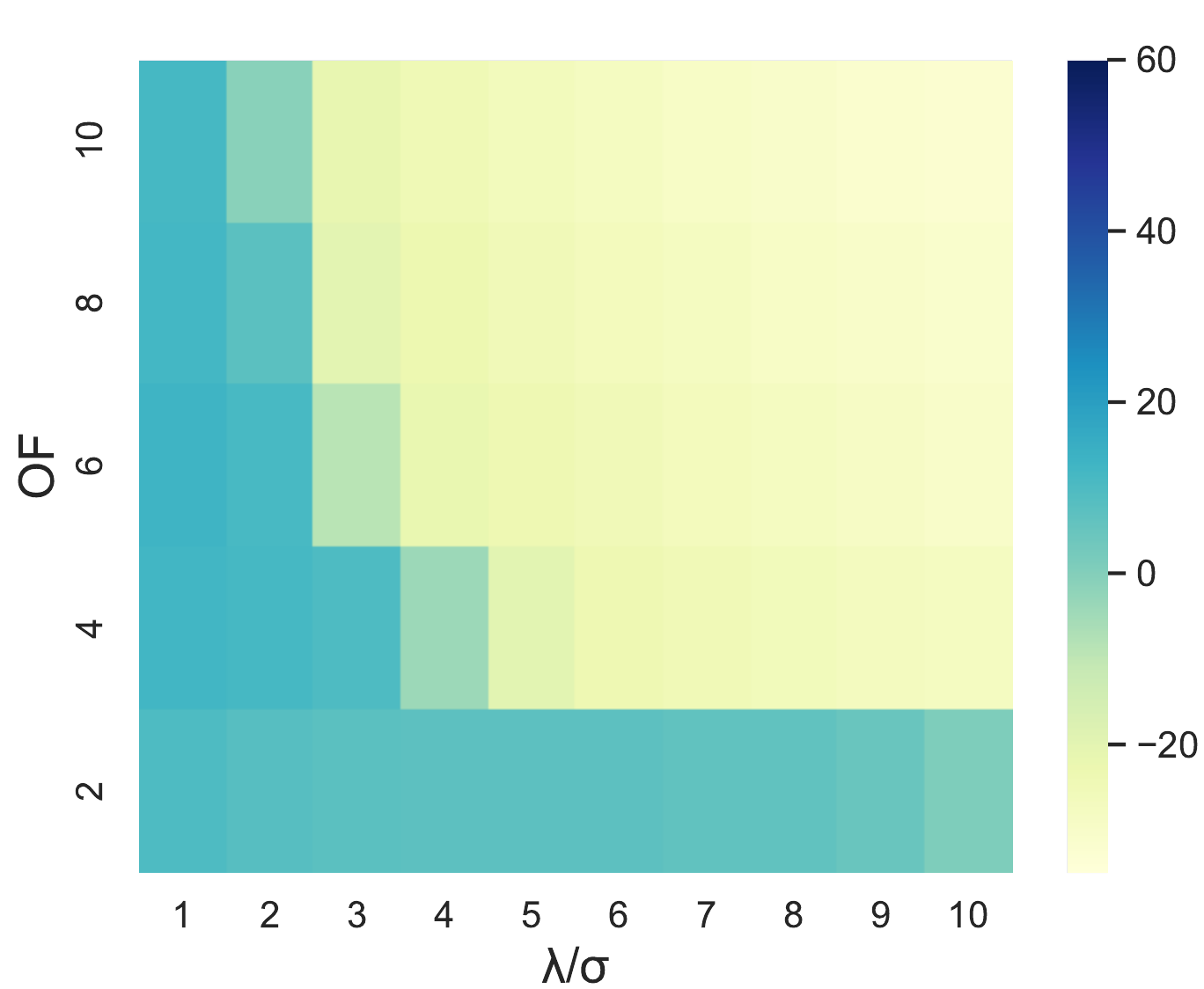}}
        \caption{Comparison of algorithms (with bounded noise) in terms of MSE when recovering a bandlimited signal from modulo samples with  $\lambda$ = 0.2. For a given SNR and OF, $B^2R^2$ has lowest MSE.}
    \label{fig:bounded_lambda_02}
    \end{figure}

     \begin{figure}[!h]
        \centering
        \begin{tabular}{cc}
           \subfigure[HOD ]{\includegraphics[width=1.6 in]{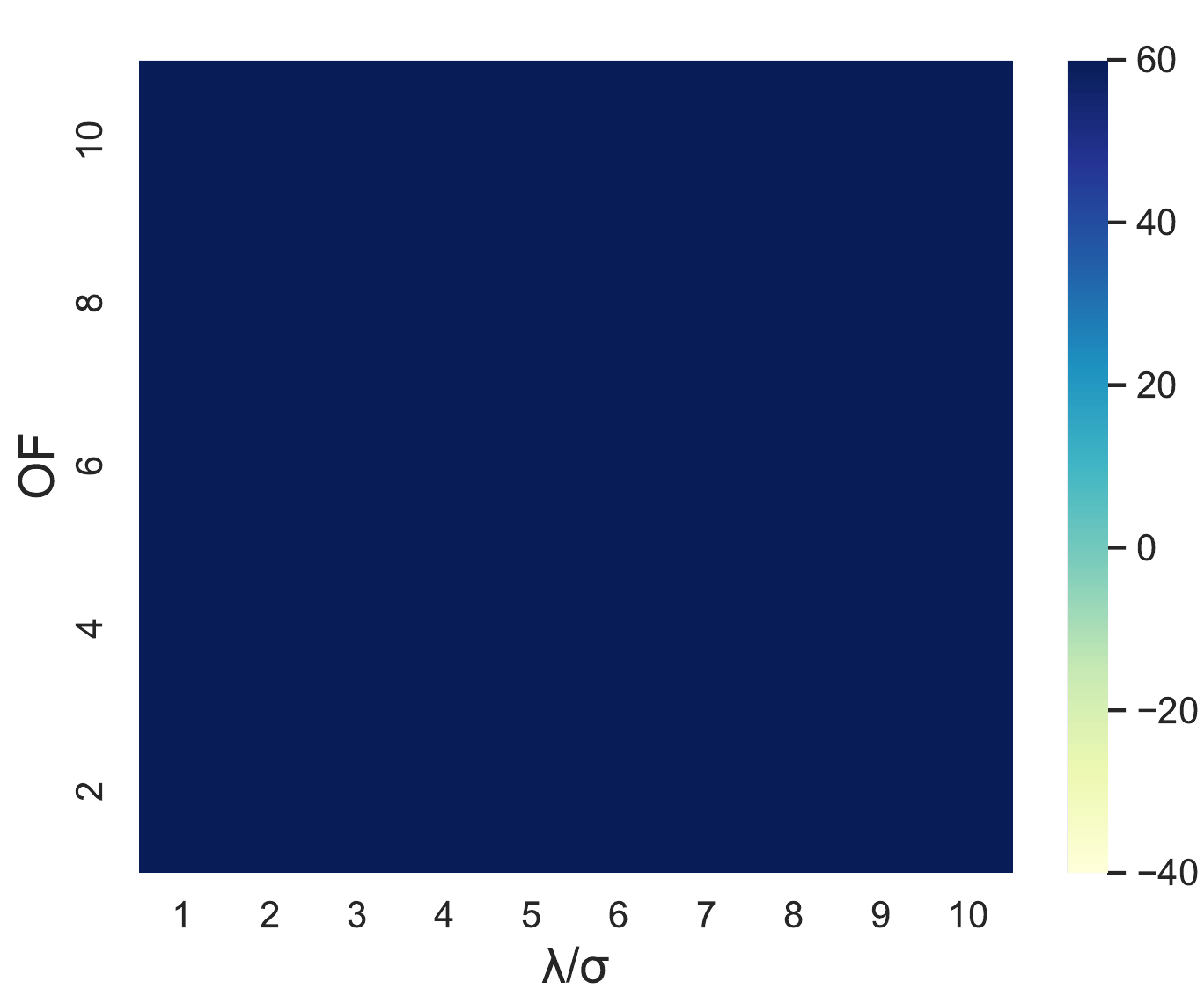}}  &  \subfigure[CPF]{\includegraphics[width=1.6 in]{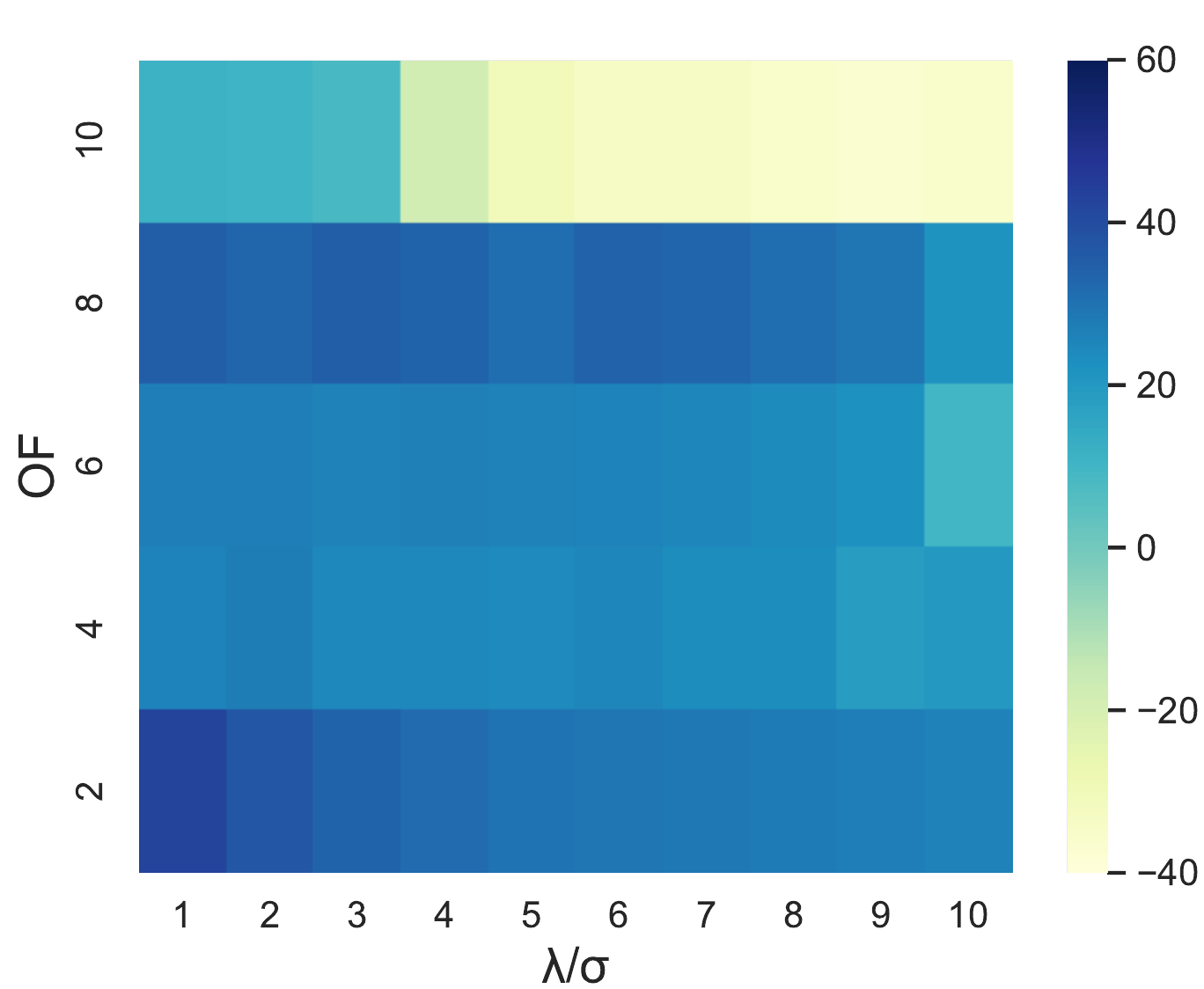}}
        \end{tabular}
        \subfigure[$B^2R^2$]{\includegraphics[width=1.6 in]{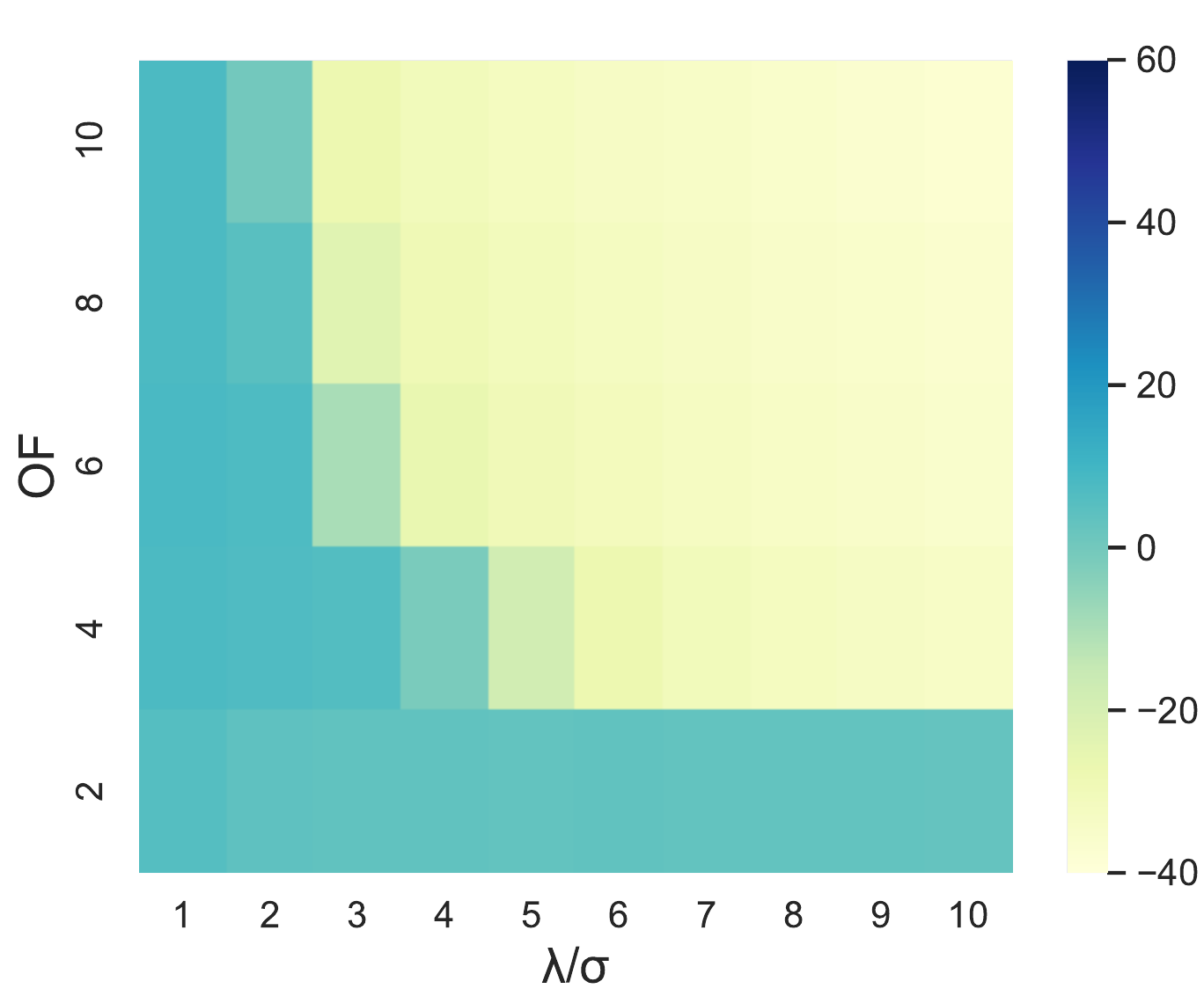}}
        \caption{Comparison of algorithms (with bounded noise) in terms of MSE when recovering a bandlimited signal from modulo samples with  $\lambda$ = 0.1. For a given SNR and OF, $B^2R^2$ has lowest MSE.}
    \label{fig:bounded_lambda_01}
    \end{figure}
    
    % \begin{figure*}[!h]
    %     \centering
    %     \subfigure[HOD ]{\includegraphics[width=2.2 in]{figs/Journal/MSE_2D/Bounded/ba_0.1.pdf}}
    %     \subfigure[CPF]{\includegraphics[width=2.2 in]{figs/Journal/MSE_2D/Bounded/chb_0.1.pdf}}
    %     \subfigure[$B^2R^2$]{\includegraphics[width=2.2 in]{figs/Journal/MSE_2D/Bounded/bbrr_0.1.pdf}}
    %     \caption{Comparison of algorithms in terms of MSE while recovering signal from modulo samples with  $\lambda$ = 0.1. For a given SNR and OF, $B^2R^2$ has lowest MSE.}
    % \label{fig:bounded_lambda_01}
    % \end{figure*}
    
    \begin{figure}[!h]
        \centering
        \begin{tabular}{cc}
           \subfigure[HOD ]{\includegraphics[width=1.6 in]{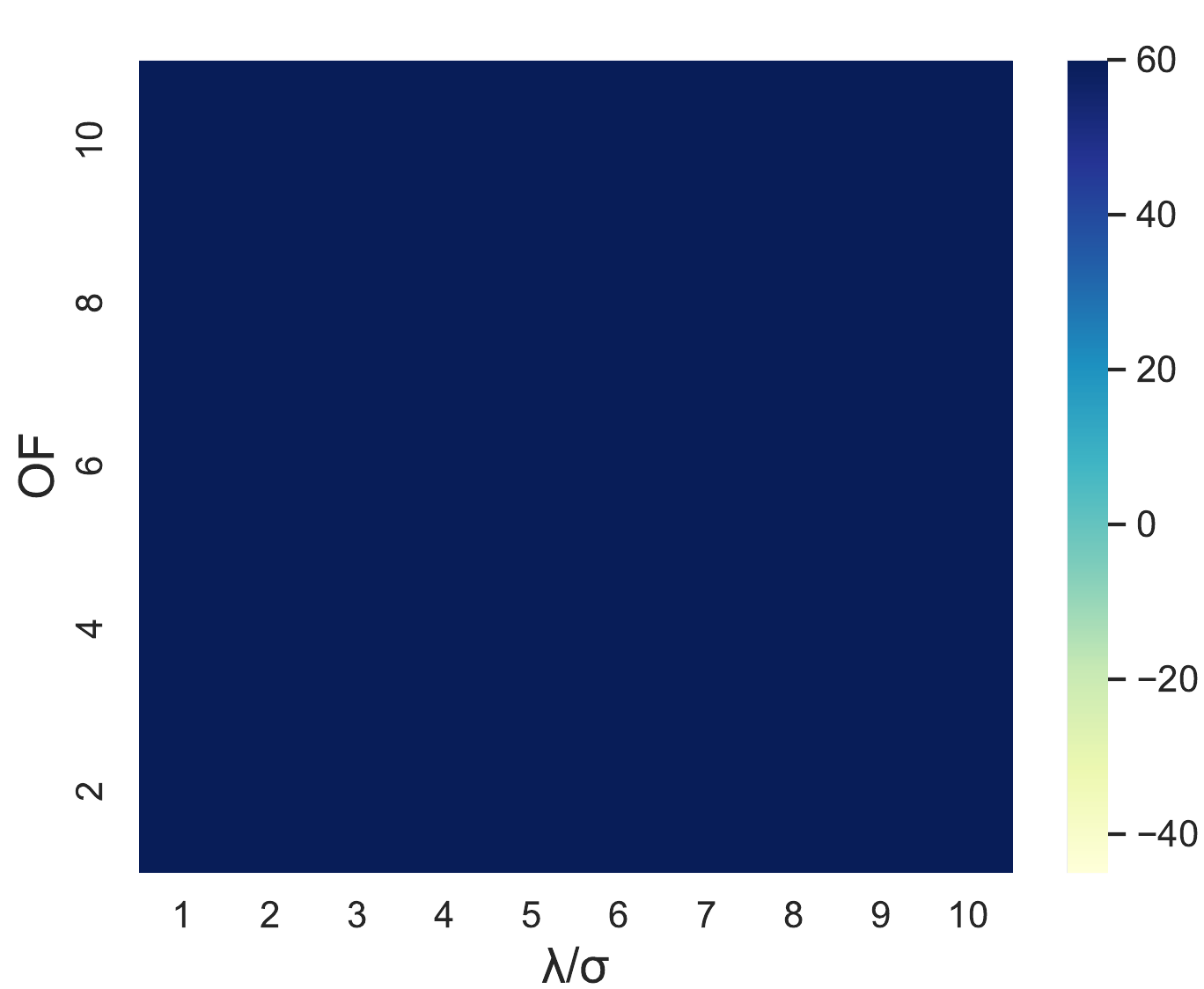}}  &  \subfigure[CPF]{\includegraphics[width=1.6 in]{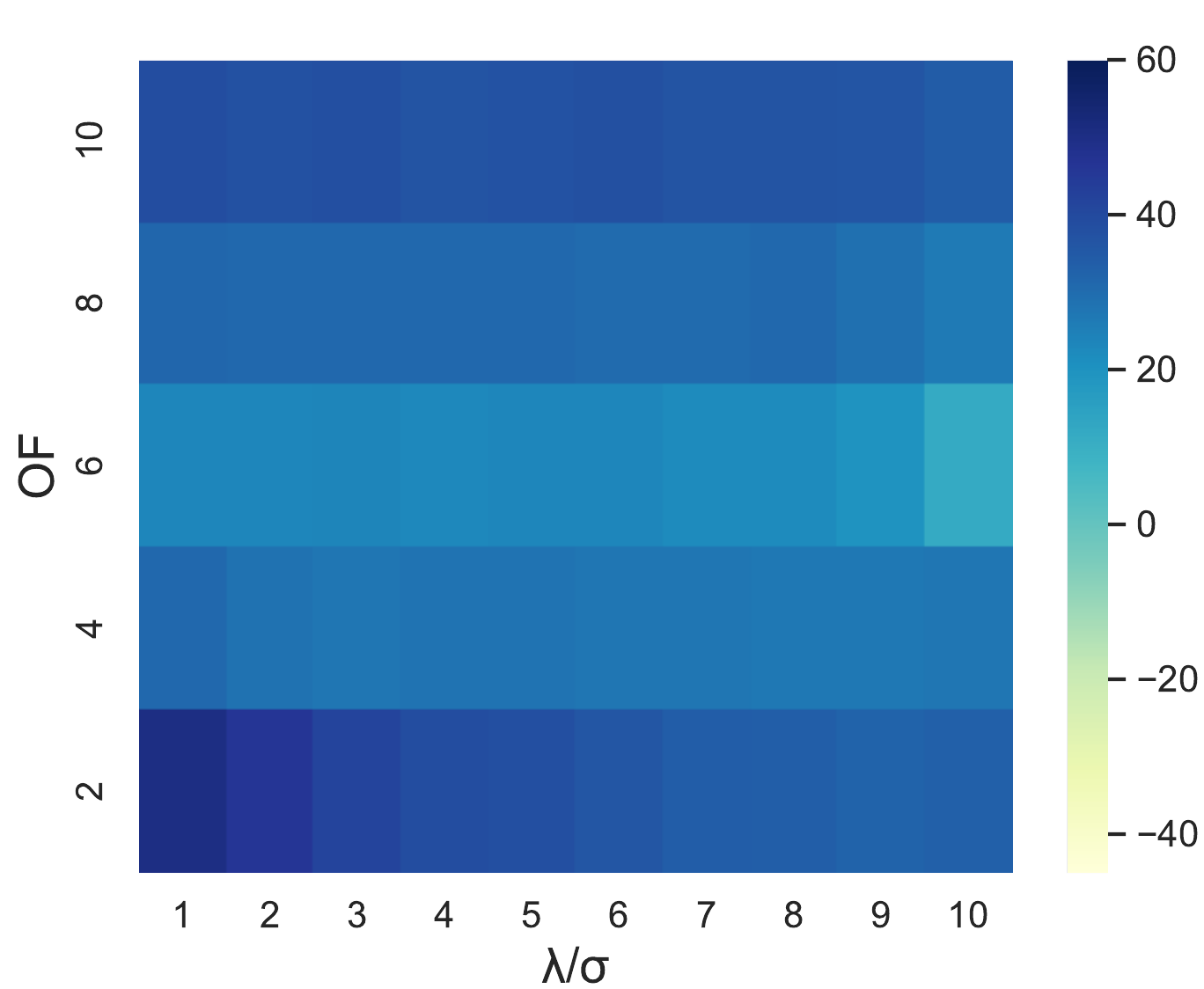}} 
        \end{tabular}
        \subfigure[$B^2R^2$]{\includegraphics[width=1.6 in]{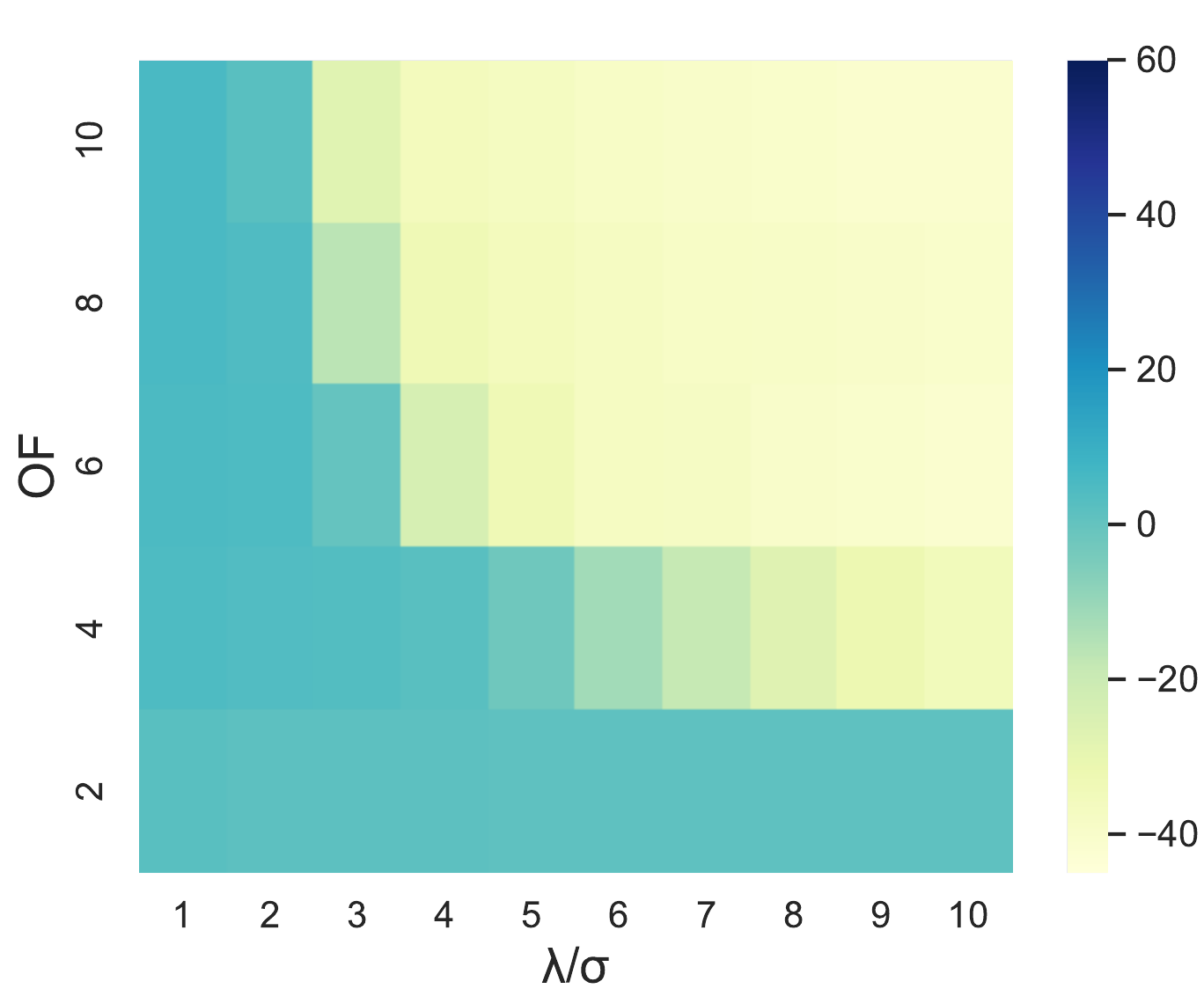}} 
        \caption{Comparison of algorithms (with bounded noise) in terms of MSE when recovering a bandlimited signal from modulo samples with  $\lambda$ = 0.05. For a given SNR and OF, $B^2R^2$ has lowest MSE.}
    \label{fig:bounded_lambda_005}
    \end{figure}
    
    \begin{figure}
    	    \begin{center}
            	{\includegraphics[width=3.2in]{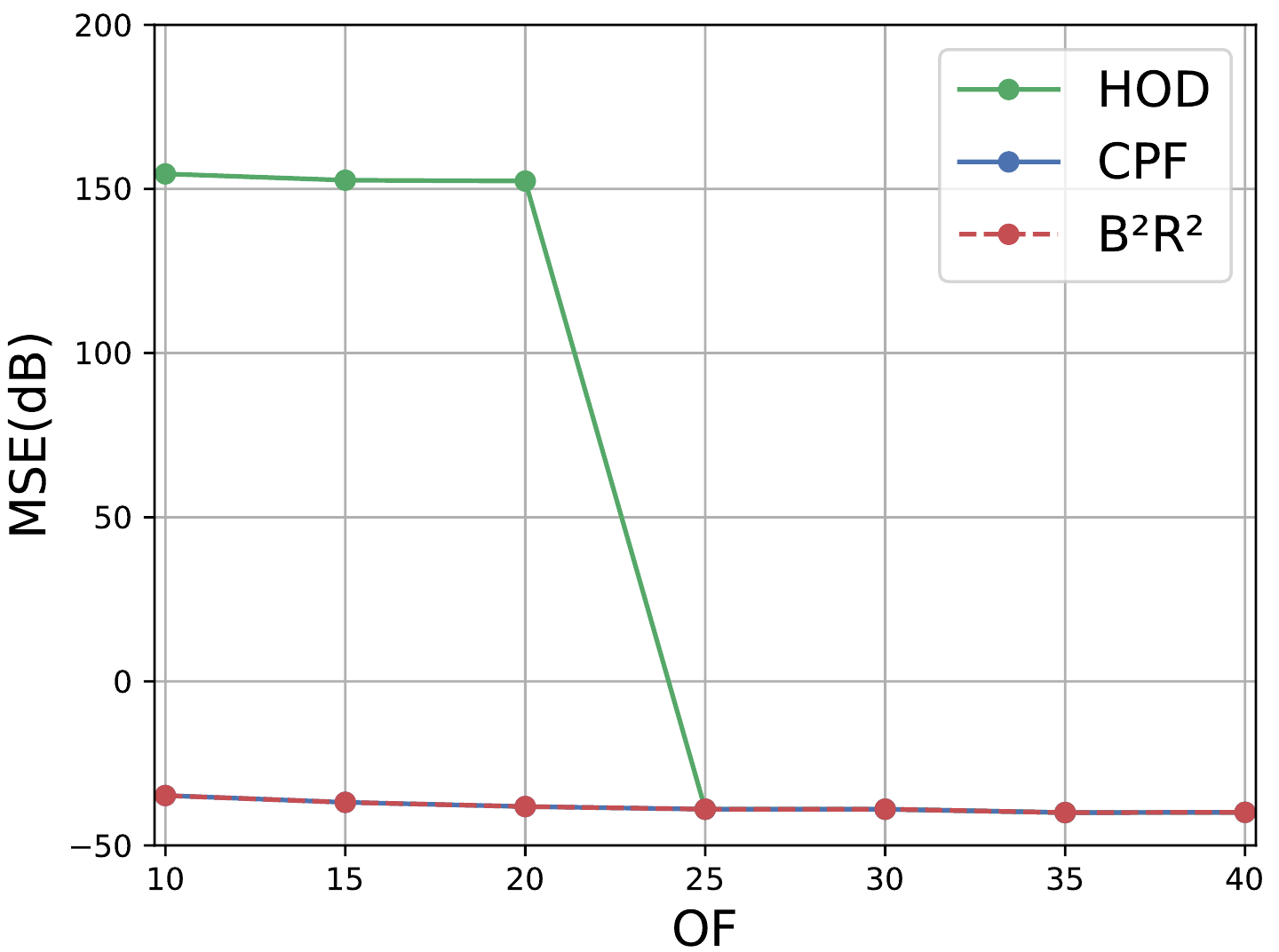}} 
        		%\end{array}
        		\caption{ Comparison of algorithms in terms of MSE in recovering a bandlimited signal from uniform noisy modulo samples with $\lambda=0.1$, and $\lambda / \sigma = 10$; The higher-order difference approach has error of $-40$ dB for
                OF $\geq 25$ whereas the remaining methods achieve $-40$ dB error for $\text{OF} = 10$.}
        		\label{fig:OF_comparison}
        	\end{center}
    \end{figure}

    % \begin{figure*}[!h]
    %     \centering
    %     \subfigure[HOD ]{\includegraphics[width=2.2 in]{figs/Journal/MSE_2D/Bounded/ba_0.05.pdf}}
    %     \subfigure[CPF]{\includegraphics[width=2.2 in]{figs/Journal/MSE_2D/Bounded/chb_0.05.pdf}} 
    %     \subfigure[$B^2R^2$]{\includegraphics[width=2.2 in]{figs/Journal/MSE_2D/Bounded/bbrr_0.05.pdf}} 
    %     \caption{Comparison of algorithms in terms of MSE while recovering signal from modulo samples with  $\lambda$ = 0.05. For a given SNR and OF, $B^2R^2$ has lowest MSE.}
    % \label{fig:bounded_lambda_005}
    % \end{figure*}
    
    \subsection{Presence of noise}
        Next, we assess the performance of $B^2 R^2$ algorithm as a function of OF, $\lambda$, and noise level. Here we focus on only the modulo operator as it enables us to compare $B^2 R^2$ method with the recently published algorithms \cite{unlimited_sampling17, uls_tsp, uls_romonov}. Specifically, we compare $B^2 R^2$ algorithm with the \emph{higher-order differences} (HOD) approach \cite{unlimited_sampling17, uls_tsp} and \emph{Chebyshev polynomial filter}-based (CPF) method \cite{uls_romonov}. We examine reconstruction problem from the following noisy measurements  
        \begin{align}
            \tilde{f}_{\lambda}(n T_s) = {f}_{\lambda}(n T_s) + v(n T_s) = \mathcal{M}_{\lambda}f(n T_s) + v(n T_s),
        \label{eq: noisy_samples}
        \end{align}
        where $v(n T_s)$ denotes noise. In the experiments we normalize the bandlimited signals to have maximum amplitude of one. In the simulations SNR is computed as $\text{SNR} = 20 \log\left(\frac{||f_{\lambda}(n T_s)||}{||v(n T_s)||}\right)$. The reconstruction accuracy of different algorithms is compared in terms of normalized mean-squared error (MSE) as $ \frac{\sum |f(nT_s) - \hat{f}(n Ts)|^2}{\sum |f(nT_s)|^2}$, where $\hat{f}(n T_s)$ denotes the estimate of $f(n T_s)$. For each noise level, 1000 independent noise realizations were generated and average MSE is computed for them. In all experiments we consider a synthetic bandlimited signal of length 1024. The structure of the generated signals is sum of sinc function with random coefficients. We examine both bounded and unbounded noises.

    \begin{figure}[]
    	    \begin{center}	{\includegraphics[width=3.2in]{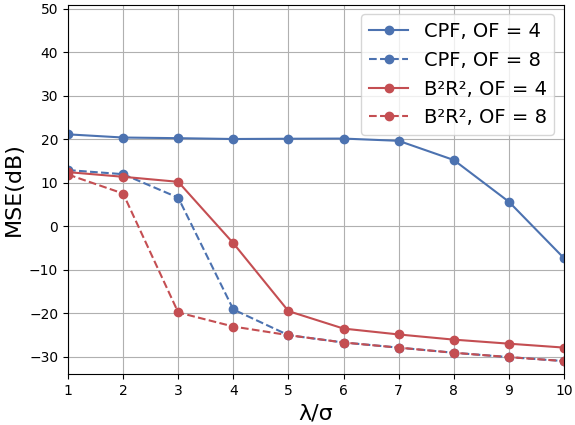}}
        		%\end{array}
        		\caption{Comparison of CPF and $B^2R^2$ algorithms (with bounded noise) in terms of MSE when recovering a bandlimited signal from modulo samples with $\lambda = 0.2$ and $\text{OF} = 4,8$. For a given ratio $\lambda /\sigma$, $B^2R^2$ has the lowest MSE.}
        		\label{fig:bound_of4}
        	\end{center}
    \end{figure}

    \begin{figure}[!h]
            \centering
    \begin{tabular}{c c}
         \subfigure[HOD ]{\includegraphics[width=1.6 in]{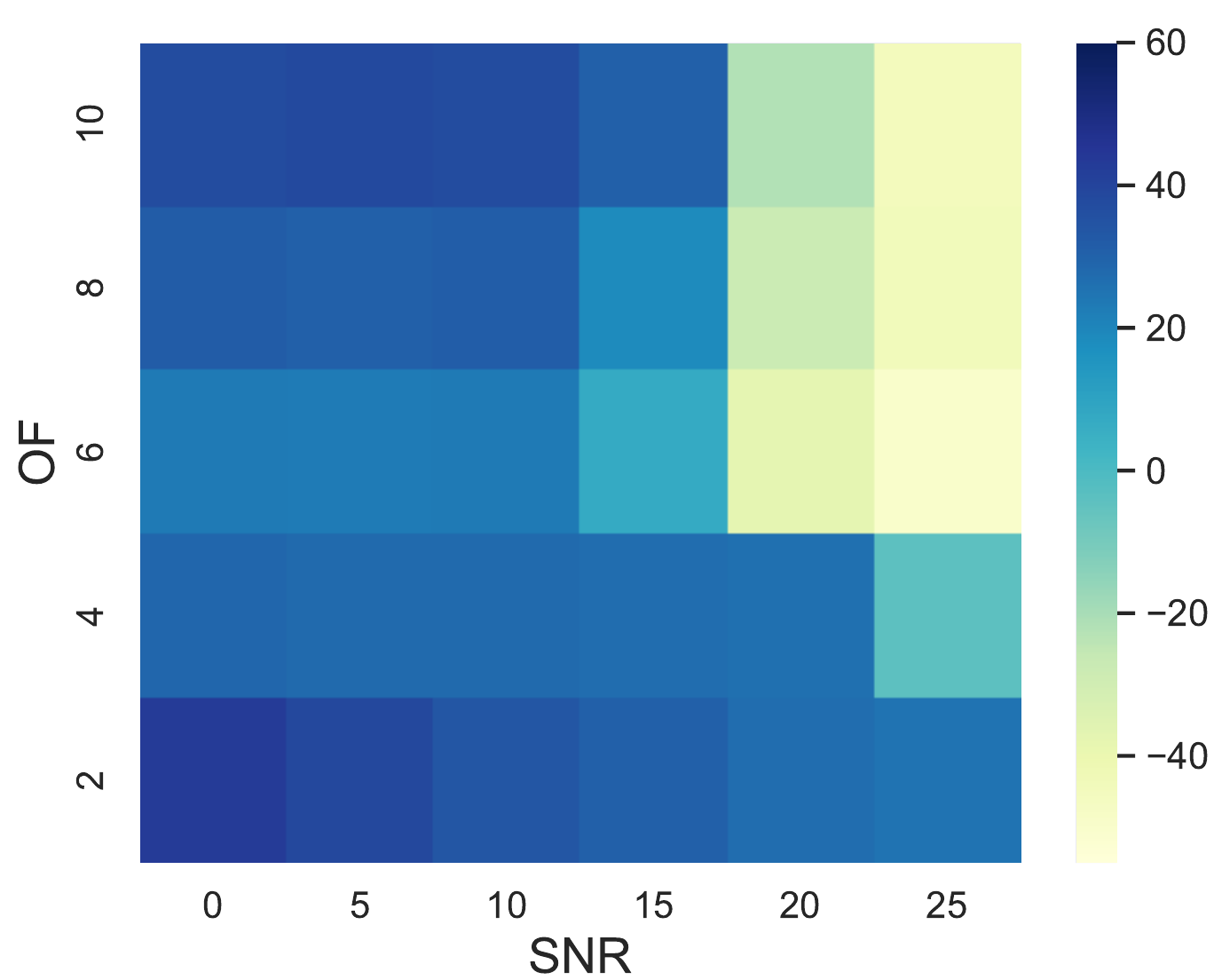}} & \subfigure[CPF]{\includegraphics[width=1.6 in]{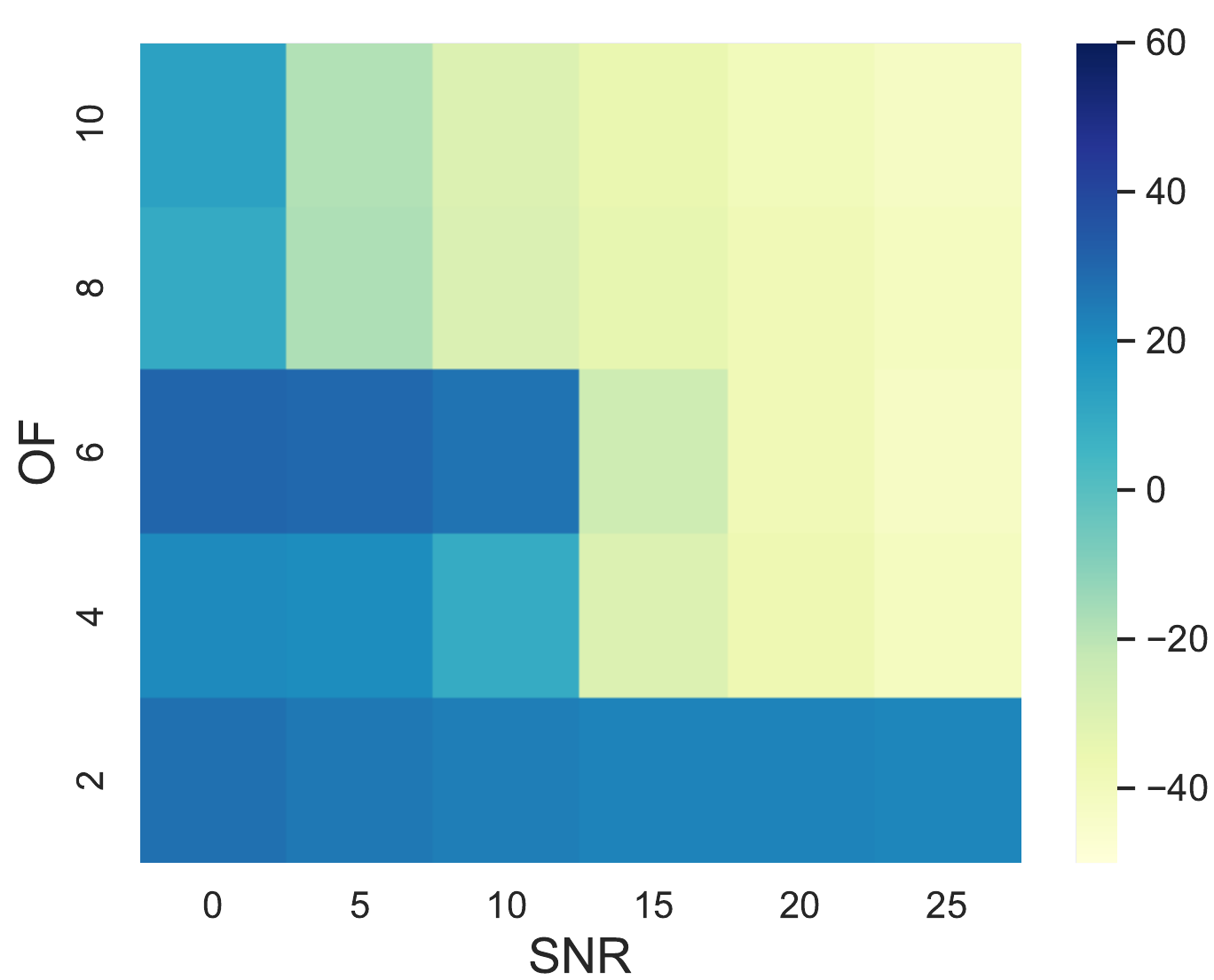}} 
    \end{tabular}
        \subfigure[$B^2R^2$]{\includegraphics[width=1.6 in]{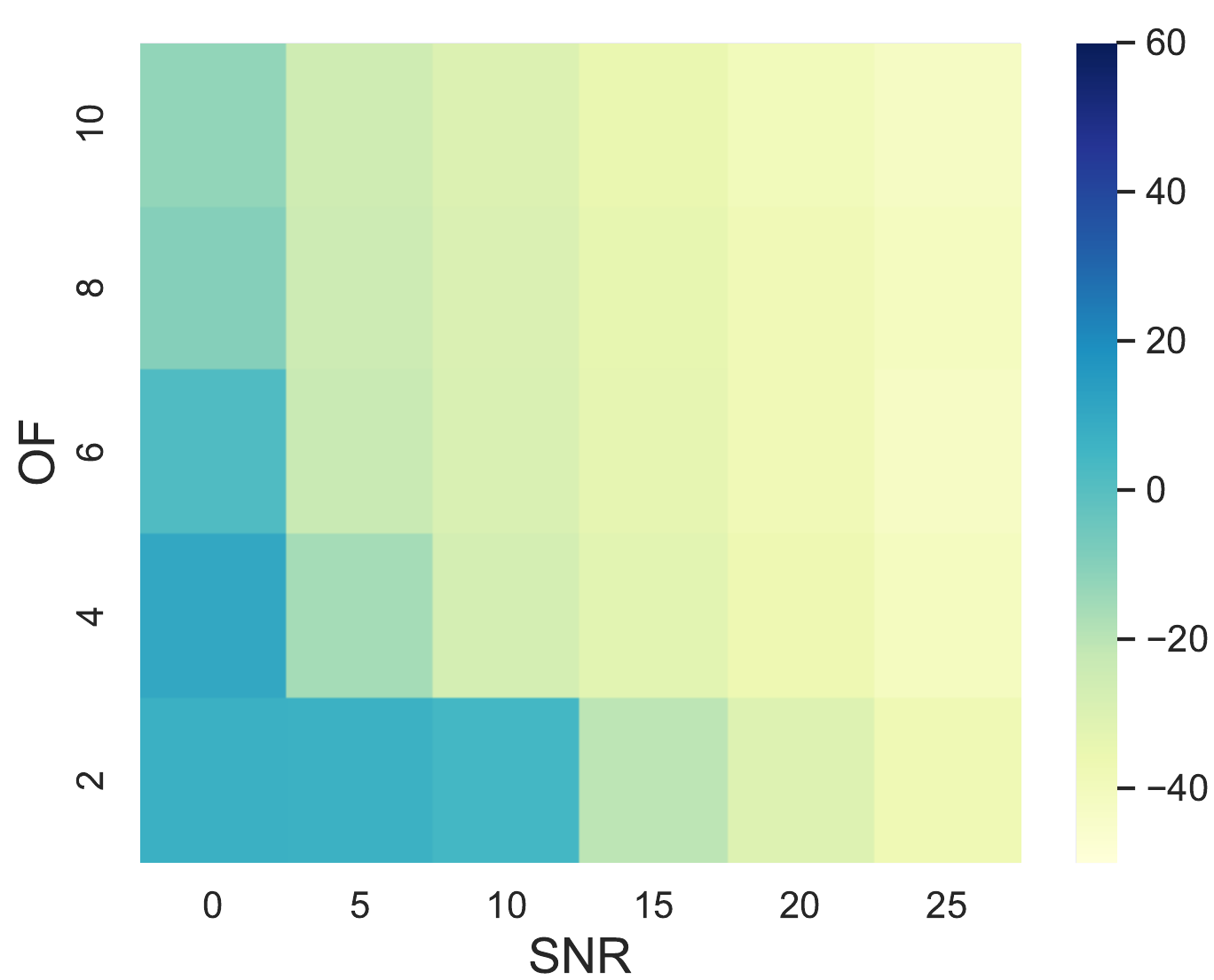}}
        \caption{Comparison of algorithms (with unbounded noise) in terms of MSE when recovering a bandlimited signal from modulo samples with  $\lambda = 0.2$. For a given SNR and OF, $B^2R^2$ has lowest MSE.}
    \label{fig:lambda_02}
    \end{figure}

    \begin{figure}[!h]
        \centering
        \begin{tabular}{c c}
           \subfigure[HOD ]{\includegraphics[width=1.6 in]{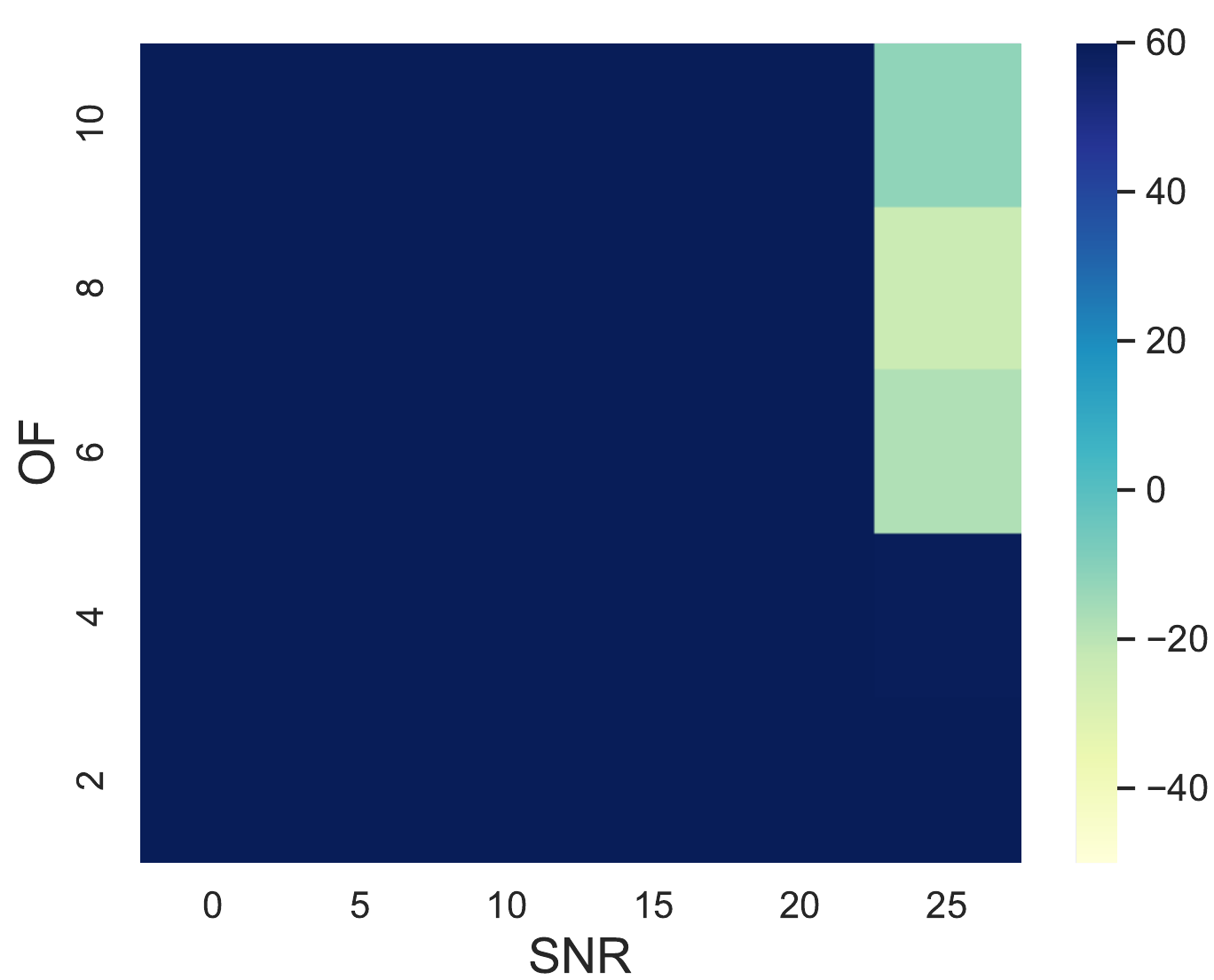}}  & \subfigure[CPF]{\includegraphics[width=1.6 in]{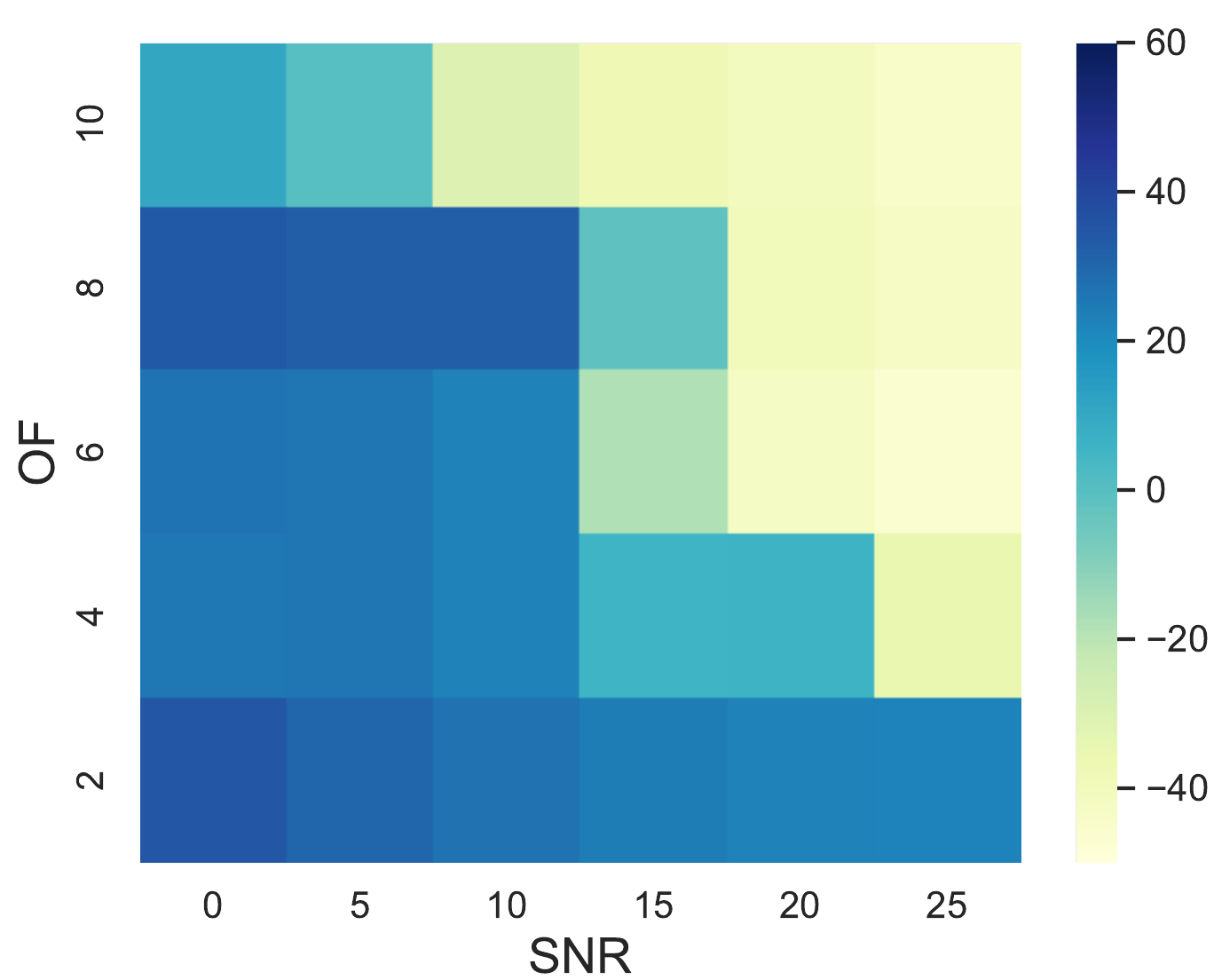}} 
        \end{tabular}
        \subfigure[$B^2R^2$]{\includegraphics[width=1.6 in]{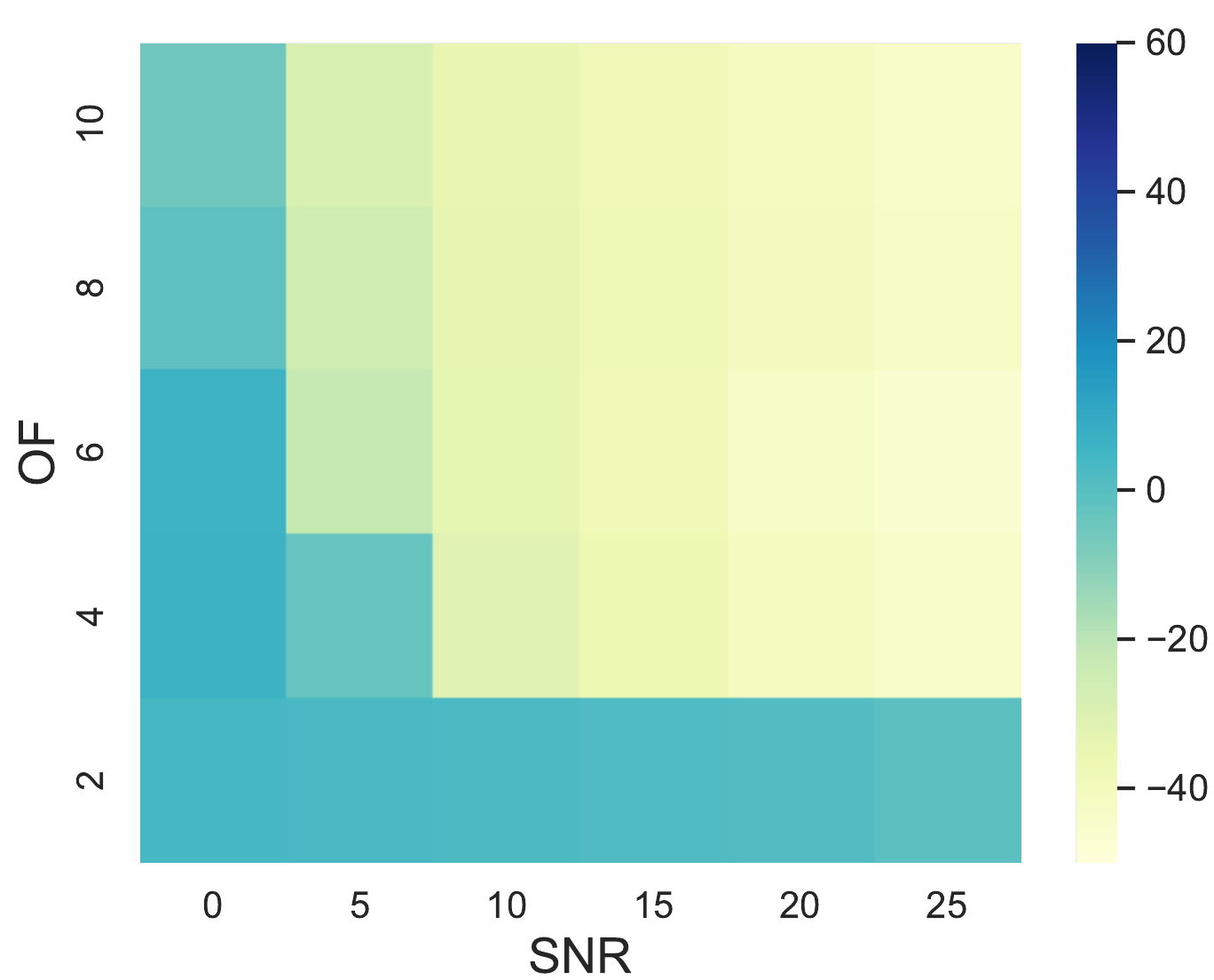}}
        \caption{Comparison of algorithms (with unbounded noise) in terms of MSE when recovering a bandlimited signal from modulo samples with  $\lambda = 0.1$. For a given SNR and OF, $B^2R^2$ has lowest MSE.}
    \label{fig:lambda_01}
    \end{figure}

    \begin{figure}[!h]
        \centering
        \begin{tabular}{c c}
           \subfigure[HOD ]{\includegraphics[width=1.6 in]{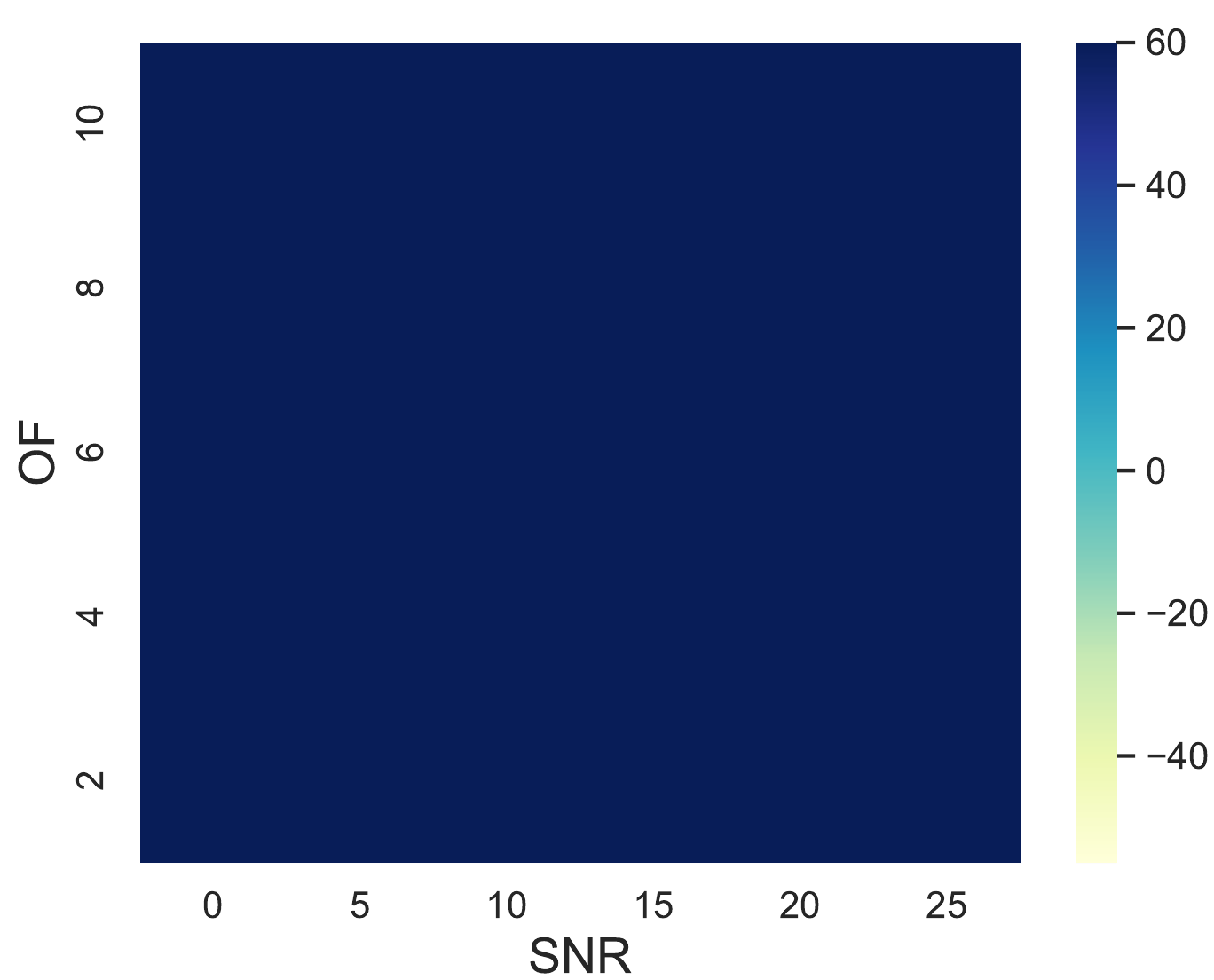}}  &  \subfigure[CPF]{\includegraphics[width=1.6 in]{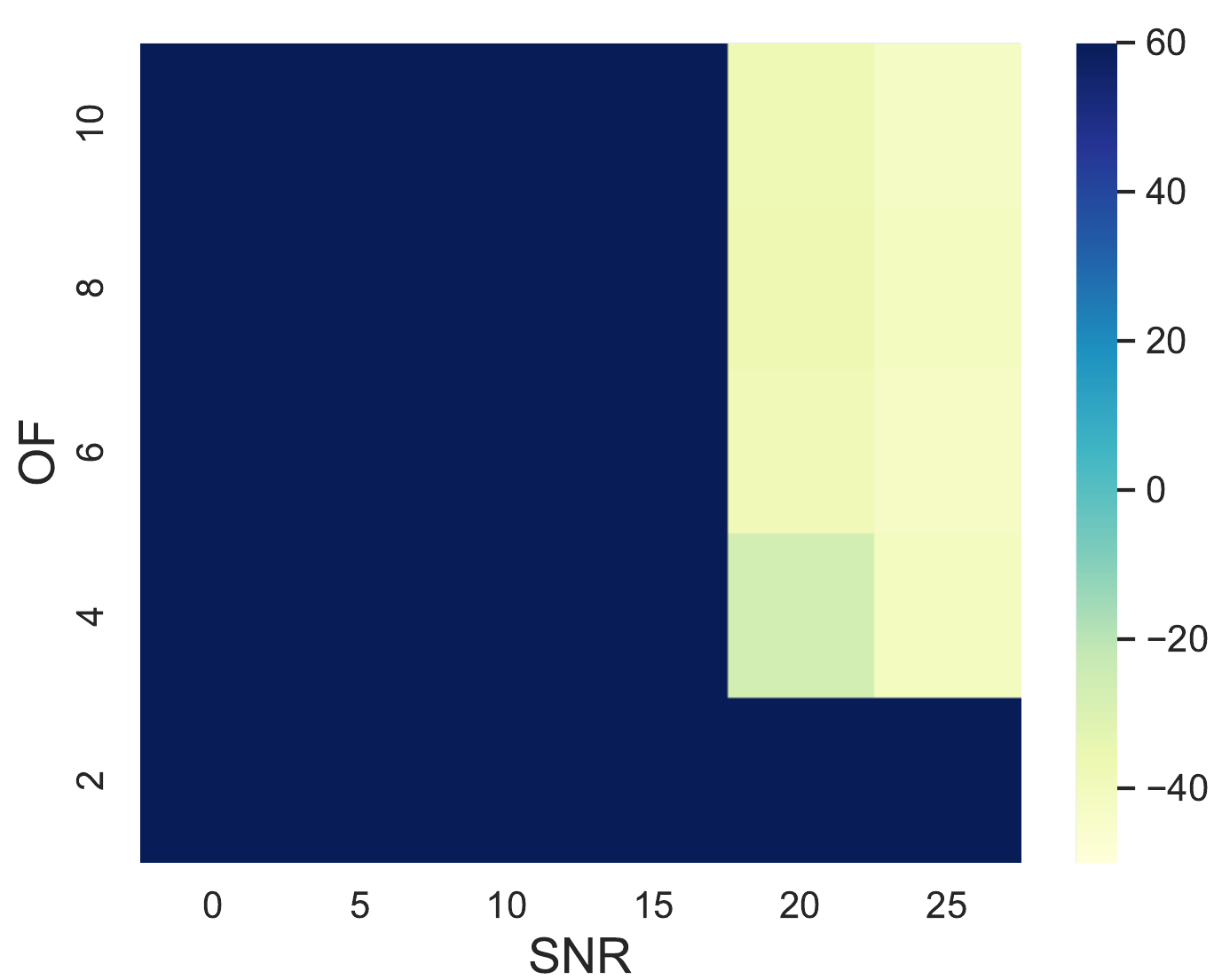}} 
        \end{tabular}
        \subfigure[$B^2R^2$]{\includegraphics[width=1.6 in]{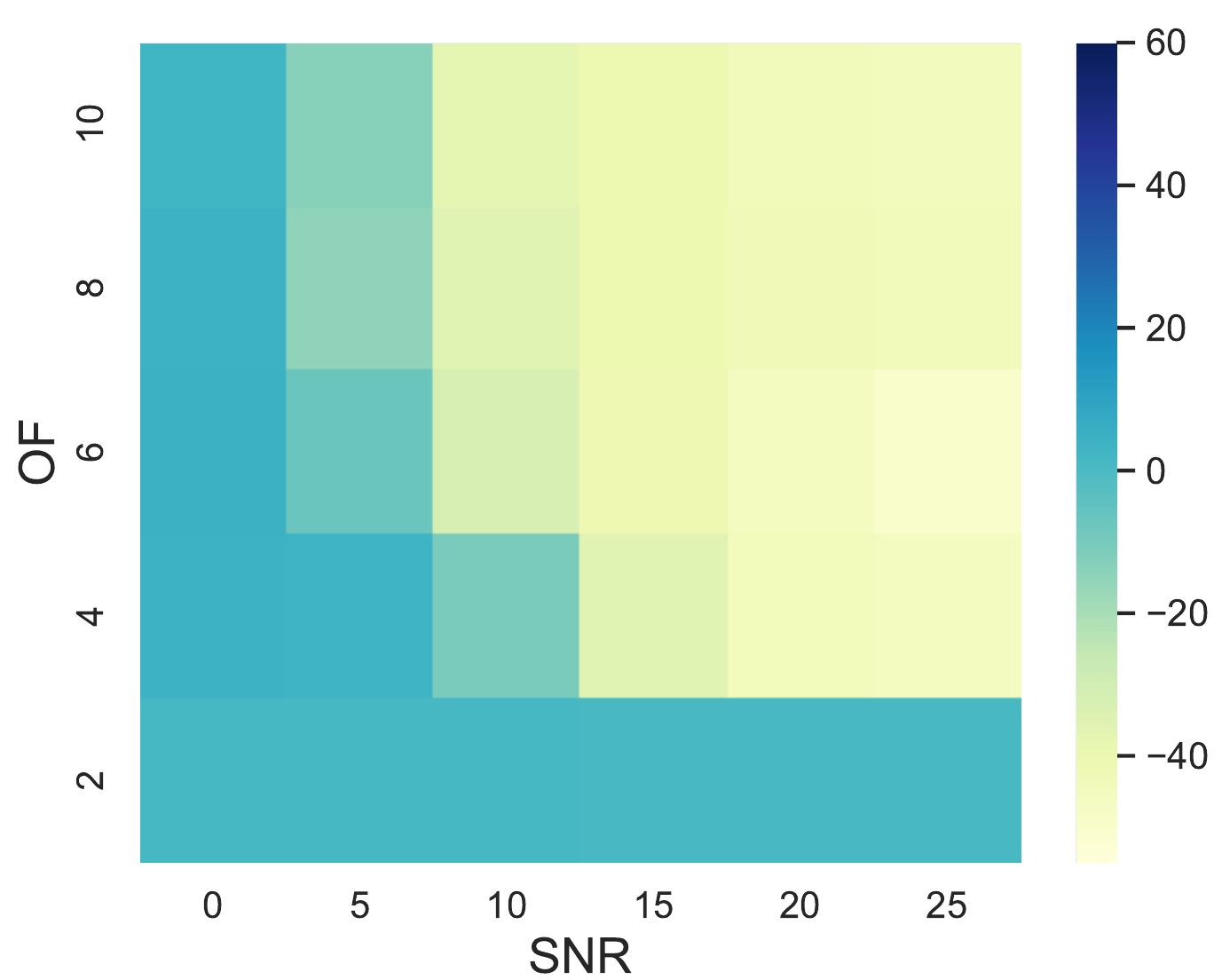}} 
        \caption{Comparison of algorithms (with unbounded noise) in terms of MSE when recovering a bandlimited signal from modulo samples with  $\lambda = 0.05$. For a given SNR and OF, $B^2R^2$ has lowest MSE.}
    \label{fig:lambda_005}
    \end{figure}

    \begin{figure}[!t]
    	    \begin{center}	{\includegraphics[width=3.2in]{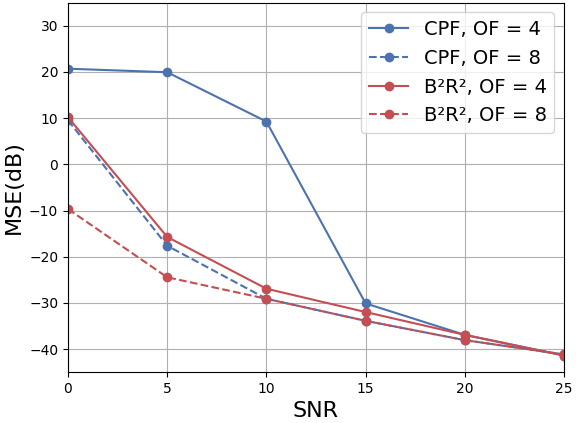}} 
        		%\end{array}
        		\caption{Comparison of CPF and $B^2R^2$ algorithms (with unbounded noise) in terms of MSE while recovering signal from modulo samples with $\lambda = 0.2$ and $\text{OF} = 4,8$. For a given SNR, $B^2R^2$ has lowest MSE.}
        		\label{fig:Unbound_of4}
        	\end{center}
    \end{figure}
 \subsubsection{Bounded noise}
    For bounded noise we assume that the noise is uniformly distributed with zero mean and $|v(nT_s)| \leq \sigma$. We compare the algorithms for different SNRs and OFs with fixed $\lambda$. Fig. ~\ref{fig:bounded_lambda_02},  \ref{fig:bounded_lambda_01}, and \ref{fig:bounded_lambda_005} show MSEs of the algorithms for $\lambda = 0.2, 0.1,$ and $0.05$, respectively.
    
     We observe that the HOD method is unable to reconstruct the signals (with MSE on the order of 60 dB) for noise levels and $\lambda$s considered in the simulations. This is because a sufficient condition for the HOD algorithm to recover the signal in the absence of noise is that $\text{OF}\geq 17$. In the presence of noise, a larger amount of oversampling is required, and hence, in this simulation setting, where $\text{OF}\leq 10$, the method fails. Both $B^2 R^2$ and CPF algorithms reconstruct the signal with lower MSEs. To ascertain the claim, we perform simulations for $\text{OF}\geq 10$ with $\lambda = 0.1$ and $\lambda / \sigma = 10$. The MSEs for the three algorithms are shown in Fig.~\ref{fig:OF_comparison}. We observe that for $\text{OF}\geq 25$, the HOD method can reconstruct the signal in this particular setting, whereas, as expected, both $B^2 R^2$ and CPF methods reconstruct the signal for lower OFs. 
    
     Comparing the $B^2 R^2$ and CPF methods, we observe that $B^2 R^2$ results in lower MSE. For a better visualization, comparison of MSEs of these two algorithms for $\text{OF} = 4$ and 8 in Fig. \ref{fig:bound_of4}. For $\text{OF} = 4$, $B^2R^2$ algorithm has $10-40$ dB lower MSE compared to CPF approach for different noise levels.

     \subsubsection{Unbounded noise}
        In these experiments, we assume that the noise samples $v(nT_s)$ are independent and identically distributed Gaussian random variables with zero mean. The variance is set to achieve a desired SNR. We compare the algorithms for different SNRs and OFs with fixed $\lambda$. Fig.~\ref{fig:lambda_02},  \ref{fig:lambda_01} and \ref{fig:lambda_005} show MSE of the algorithms for $\lambda = 0.2,\,0.1,\,0.05$, respectively. As in the case of bounded noise, the HOD method is unable to recover the signal for this experimental setup. Comparing the rest of the methods, the $B^2R^2$ algorithm results in lower MSE than that of the CPF method for a given $\lambda$, MSE, and OF.
        
        For a better visualization, we compared the $B^2 R^2$ and CPF methods in terms of MSE for $\text{OF} = 4$ and $\text{OF} = 8$ in Fig. \ref{fig:Unbound_of4}. 
        For $\text{OF} = 4$, we note that for low SNR values, the $B^2R^2$ algorithm results in $10-30$ dB lower MSE when compared to the CPF approach.

\section{conclusion}
We propose a nonlinear operator that can be used to address the dynamic range issue of ADCs. The proposed is as a generalization of existing operators such as companding and modulo. We show that bandlimited signals can be perfectly reconstructed from the samples of the proposed nonlinear operator provided that the sampling rate is greater than the Nyquist rate. We also propose a robust algorithm to recover the true samples from the nonlinear samples. Our results show that our algorithm operates at lower sampling rate compared to the existing approaches for different noise levels and dynamic ranges.

% \newpage
\appendix

   In this appendix, we provide proof of Theorem \ref{theorem:necessary_identifiability}. We first present the proof for the necessary part and then discuss sufficiency.\\
   
   \label{proof:necessary_identifiability}
   \noindent \emph{Necessary Part:} Let $\omega_s = \frac{2\pi}{T_s}$ denote the sampling rate in rad/sec. 
    To prove the necessary condition on the sampling rate, we consider two cases: (i) Sampling below the Nyquist rate: $\omega_s = \frac{2\pi}{T_s} < 2\omega_m$ and (ii) Sampling at the Nyquist rate: $\omega_s = 2\omega_m$. When the bandlimited signal $f(t)$ is sampled below the Nyquist rate then there exists another signal $\hat{f}(t) \in L^2(\mathbb{R})\cap B_{\omega_{m}}$ such that $f(nT_s) = \hat{f}(nT_s), $ for all $ n \in \mathbb{Z}$ due to aliasing. This implies that the samples of the framework shown in Fig.~\ref{fig:sampling_block} are the same for inputs $f(t)$ and $\hat{f}(t)$, that is,  $f_{\lambda}(nT_s) = \hat{f}_{\lambda}(nT_s), $ for all $ n \in \mathbb{Z}$. Hence the signal $f(t)$ is not uniquely identifiable from $f_{\lambda}(nT_s)$ for $\omega_s < 2\omega_m$.

    Next, consider sampling at the Nyquist rate $\omega_s = 2\omega_m$. We show that there exists a signal $f(t)$ which is not uniquely identifiable. In other words, given $f(t)$ there exists another bandlimited signal $\hat{f}(t) \in L^2(\mathbb{R})\cap B_{\omega_{m}}$ such that $f_{\lambda}(nT_s) = \hat{f}_{\lambda}(nT_s), $ for all $ n \in \mathbb{Z}$. Consider an $f(t) \in L^2(\mathbb{R})\cap B_{\omega_{m}}$ such that $|f(n_0 T_s)|>\lambda$ for some $n_0 \in \mathbb{Z}$ . For example, for $f(t) = 2 \lambda \,\text{sinc} \left(\frac{t-n_0 T_s}{T_s}\right)$ satisfies the mentioned sampling conditions. 
    % The condition $|f(n_0 T_s)|>\lambda$ is ensures that the assumption made in Section~\ref{sec:problem_formulation} is satisfied. 
    We construct $\hat{f}(t)$, from the samples $f(nT_s)$, by defining its Nyquist rate samples as
    \begin{align}
    \label{eq:construct_sequence}
        \hat{f}(nT_s) = \begin{cases}
            g^{-1} \circ \mathcal{G}_{\lambda}f(n T_s), & n = n_0,\\
             f(n T_s), & \text{otherwise},
        \end{cases}
    \end{align}
    where $g^{-1}(\cdot)$ is the inverse of $g(\cdot)$ (cf. \eqref{eq:non-linear-operator}). 
    The range and domain of the functions $g(\cdot), \, g^{-1}(\cdot)$ are given by the interval $[-\lambda, \lambda]$.
    % as $g: [-\lambda, \lambda] \to [-\lambda, \lambda]$, that is, 
    % \begin{align}
    %     g^{-1}: [-\lambda, \lambda] \to [-\lambda, \lambda].
    %     \label{eq:ginverse}
    % \end{align} 
    Since $|\mathcal{G}_\lambda f(n_0 T_s)|\leq \lambda$ (cf. \eqref{eq:bounded_g}), then from the aforementioned range of $g^{-1}(\cdot)$, we infer that 
    \begin{align}
        |\hat{f}(n_0 T_s)| \leq \lambda.
        \label{eq:bounded_fhat}
    \end{align}
    Hence $f(n_0 T_s) \neq \hat f(n_0 T_s)$ and thus $f(t) \neq \hat{f}(t)$. 
    % Further, as $|f(n T_s)|\leq \lambda$ for $n \in \mathbb{Z} \backslash \{n_0\}$, from the definition of $\mathcal{G}_\lambda$, we have that $f(n T_s) = \hat f(n T_s)$ for $n \in \mathbb{Z} \backslash \{0\}$.
    
    The signal $\hat{f}(t)$ is in $B_{\omega_{m}}$ by construction. 
    Next, we show that it is also in $L^2(\mathbb{R})$. From Parseval's formula we have that 
    \begin{align}
        \int\limits_{-\infty}^{\infty} |\hat{f}(t)|^2 \mathrm{d}t  = T_s \sum_{n \in \mathbb{Z}} & |\hat{f}(nT_s)|^2 , \\
        = T_s \sum_{n \in \mathbb{Z}} |f(nT_s)|^2 & + T_s|\hat{f}(n_0 T_s)|^2 - T_s |f(n_0 T_s)|^2.
    \end{align}
    Since $f(t) \in L^2(\mathbb{R})$ the first and the third terms on the right-hand side are finite and from \eqref{eq:bounded_fhat} we have that $\hat{f}(t) \in \mathbb{L}^2(\mathbb{R})$.

    Next, consider the nonlinear samples of $\hat{f}(t)$. Since the operator $\mathcal{G}_{\lambda}$ is memoryless and $f(nT_s) = \hat f(nT_s), n \in \mathbb{Z}\backslash \{n_0\}$, we have that $\hat{f}_{\lambda}(nT_s) = f_{\lambda}(nT_s)$ for all $n \in \mathbb{Z}\backslash \{n_0\}$. For $n = n_0$ we have the following equalities.
    \begin{align}
        \hat{f}_{\lambda}(n_0 T_s) & = \mathcal{G}_{\lambda} \hat{f}(n_0 T_s), \nonumber\\
        & = g \circ \hat{f}(n_0 T_s), \quad \text{(from \eqref{eq:non-linear-operator} and \eqref{eq:bounded_fhat})}\nonumber\\
        & = g \circ g^{-1} \circ \mathcal{G}_{\lambda}f(n_0 T_s), \quad \text{(from \eqref{eq:construct_sequence} )}\nonumber\\
        & = \mathcal{G}_{\lambda}f(n_0 T_s) = f_{\lambda}(n_0 T_s).
    \end{align}
    This shows that there exists two different bandlimited functions $f(t)$ and $\hat{f}(t)$ whose non-linear samples are identical when measured at the Nyquist rate. This proves that it is necessary to sample above the Nyquist rate.\\

    \noindent \emph{Sufficient Part:} We prove this part by contradiction. Assume that there exist two different bandlimited signals $f_1(t), f_2(t) \in L^2(\mathbb{R})\cap B_{\omega_m}$ with the same non-linear samples, sampled above the Nyquist rate. That is,
    \begin{align}
        \label{eq:non_linear_samples_f12}
        \mathcal{G}_{\lambda}f_1(n T_s) = \mathcal{G}_{\lambda}f_2(n T_s), \, \forall n \in \mathbb{Z}.    
    \end{align}
    Since $f_1(t)$ and $f_2(t)$ are bandlimited and have finite energy, their Fourier transforms have finite energy (From Parseval's theorem) and are absolutely integral (by applying H\"older's inequality). Hence, from the Riemann–Lebesgue lemma, we have that $|f_k(t)| \rightarrow 0$ as $|t| \rightarrow \infty$ for $k = 1, 2$. In other words, for a given $\lambda$ there exist an integer $N_{\lambda}$ such that 
    \begin{align}
        |f_k(nT_s)| < \lambda, \quad \forall |n|> N_{\lambda}, \quad k =1, 2.
        \label{eq:Nlambda}
    \end{align}
    From \eqref{eq:non-linear-operator} and \eqref{eq:non_linear_samples_f12}, for $|n| > N_{\lambda}$ we have that
    \begin{align}
        g \circ f_1(nT_s) & = g \circ f_1(nT_s), \\
        \implies \, f_1(n T_s) & = f_2(nT_s),
        \label{eq:equal_Nlambda}
    \end{align}
        where the last equality is due to the fact that $g$ is invertible. 
        
        Next, consider an $\omega_m$-bandlimited function $h(t) = f_1(t) - f_2(t)$. Since $h(nT_s) = f_1(nT_s) - f_2(nT_s)$, from \eqref{eq:equal_Nlambda} we have that
        \begin{align}
            h(nT_s) = 0, \quad |n|>N_{\lambda}.
            \label{eq:hnTs=0}
        \end{align}
     The DTFT of $h(nT_s)$ is given as    
         \begin{align}
         H(e^{j\omega T_s}) &= \sum_{n={-\infty}}^{\infty}h(n T_s)e^{-\mathrm{j} n T_s \omega } \\ &= 
        \sum_{n={-N_\lambda}}^{N_\lambda}h(n T_s)e^{-\mathrm{j} n T_s \omega }\\ & = \frac{1}{T_s}\sum_{k={-\infty}}^{\infty} H\left(\omega - k\omega_s \right),
    \end{align}
    where $H(\omega)$ is the CTFT of $h(t)$. Since $H(\omega) = 0, \omega \not \in [-\omega_m, \omega_m]$ and $\omega_s > 2\omega_m$, we note that
    \begin{align}
        H(e^{\mathrm{j}\omega}) = \sum_{n={-N_\lambda}}^{N_\lambda}h(n T_s)e^{-\mathrm{j} n T_s \omega } = 0, \,\,\, \text{for} \,\,\, \omega \in [\omega_m, \omega_s/2].
    \label{eq: DTFT_zeros}
    \end{align}
    Since $H(e^{\mathrm{j}\omega})$ is a trigonometric polynomial and it is equal to zero over an interval, then by using the identity theorem \cite[Page 122]{ablowitz_fokas_2003} we have that $H(e^{\mathrm{j}\omega}) = 0$ for all $\omega \in \mathbb{R}$. This implies that $h(nT_s) = 0, $ for all $ n \in \mathbb{Z}$. Hence, $f_1(nT_s)= f_2(nT_s), $ for all $ n \in \mathbb{Z}$ and  therefore $f_1(t) = f_2(t)$, which contradicts our initial assumption. A similar lines of proof is used in \cite{uls_identifiability} to derive sufficient conditions for bandlimited signals from their modulo samples.

    % As H (ω) in (2) is a trigonometric polynomial and hence an
    % entire function, Theorem 2 (the Identity Theorem [18]) implies
    % that H (ω) and the zero function agree on the full complex
    % plane; consequently h = 0 and we obtain a contradiction to
    % our assumption in (1), that is, f 6= g

\bibliographystyle{IEEEtran}
\bibliography{US_biblios,refs}

\end{document}